\documentclass[12pt]{article}

\usepackage{graphics,graphicx,fullpage,natbib,multirow}
\usepackage{amsmath,amssymb,verbatim,epsfig}
\usepackage[dvipsnames,usenames]{color}
\usepackage{caption,subcaption,booktabs}

\newtheorem{proposition}{Proposition}

\newtheorem{defin}{\bf Definition}

\newenvironment{proof}{\noindent{\bf Proof.}}{$\diamond$}

\def\ga{\mbox{Ga}}

\def\nbi{\mbox{NB}}
\def\geo{\mbox{Geo}}

\def\po{\mbox{Po}}

\def\un{\mbox{Un}}

\def\dir{\mbox{Dir}}

\def\E{\mbox{E}}
\def\V{\mbox{Var}}

\def\Cor{\mbox{Corr}}

\def\Cr{\mbox{Corr}}
\def\Cv{\mbox{Cov}}

\def\rest{\mbox{rest}}
\def\data{\mbox{data}}
\def\ba{{\bf a}}

\def\bx{{\bf x}}
\def\by{{\bf y}}
\def\bz{{\bf z}}

\def\bX{{\bf X}}
\def\bY{{\bf Y}}
\def\bZ{{\bf Z}}

\def\simind{\stackrel{\mbox{\scriptsize{ind}}}{\sim}}

\newcommand{\balpha}{\boldsymbol{\alpha}}

\newcommand{\bgamma}{\boldsymbol{\gamma}}

\newcommand{\bpi}{\boldsymbol{\pi}}

\newcommand{\btheta}{\boldsymbol{\theta}}

\newcommand{\NB}{\mathbb{N}}
\newcommand{\RB}{\mathbb{R}}

\begin{document}

\baselineskip=24pt

\title{\bf Negative binomial models for development triangles of counts}
\author{{\sc Luis E. Nieto-Barajas \& Rodrigo Targino} \\[2mm]
{\sl ITAM, Mexico \& FGV, Brazil} \\[2mm]
{\small {\tt luis.nieto@itam.mx \& rodrigo.targino@fgv.br}} \\}
\date{}
\maketitle

\begin{abstract}
Prediction of outstanding claims has been done via nonparametric models (chain ladder), semiparametric models (overdispersed poisson) or fully parametric models. In this paper, we propose models based on negative binomial distributions for the prediction of outstanding number of claims, which are particularly useful to account for overdispersion. We first assume independence of random variables and introduce appropriate notation. Later, we generalise the model to account for dependence across development years. In both cases, the marginal distributions are negative binomials. We study the properties of the models and carry out bayesian inference. We illustrate the performance of the models with simulated and real datasets. 
\end{abstract}

\vspace{0.2in} \noindent {\sl Keywords}: claims reserving, integer-valued time series, latent variables, moving average process, stationary process.

\section{Introduction}
\label{sec:intro}

The prediction of outstanding claims is of interest for actuaries and insurance companies for claim reserving. The problem is usually referred to as incurred but not reported (IBNR) claims and can be stated as follows: Given a set of incremental (number or amount of) claims observed $\{X_{i,j},\,i=1,\ldots,n,\,j=1,\ldots,n-i+1\}$, usually represented by a run-off triangle like that in Table \ref{tab:triangle}, where the index $i$ represents the year of origin or accident year, and the index $j$ represents the development year or delay year. The idea is to predict unobserved claims $\{X_{i,j},\,i=2,\ldots,n,\,j=n-i+2,\ldots,n\}$, i.e., lower-right triangle in Table \ref{tab:triangle}, to determine the individual reserves (aggregated outstading claims) per year $N_i=\sum_{j=n-i+2}^n X_{i,j}$, for $i=2,\ldots,n$ and the total reserve $N=\sum_{i=2}^n N_i$. 

The first technique to estimate the reserve in this context is the chain-ladder. There is no official reference on who and when this technique was proposed, but it is definitely the most popular. The chain-ladder technique is not based on any specific model assumptions and therefore can be seen as a nonparametric model and it is indistinctively applied to the number of claims (discrete data) and to claim amounts (continuous data). See, for example, \cite{england&verral:02}. 

Many authors have proposed stochastic versions of the chain-ladder by assuming an underlying parametric model for incremental claims $X_{i,j}$. For example, for discrete data, \cite{renshaw&verrall:98} considered a Poisson model with a logarithmic link function for the mean of the form $\log\mu_{i,j}=\nu+\alpha_i+\beta_j$. \cite{verrall:00} also considers a Poisson model but takes a multiplicative expression to represent the mean, say $\mu_{i,j}=\alpha_i\beta_j$, where after a reparameterisation and considering a gamma conjugate prior for the row parameters $\alpha_i$, the marginal model becomes a negative binomial. 

For continuous data, \cite{kremer:82} considers a lognormal model, where the mean of the underlying normal distribution is of the form $\mu_{ij}=\nu+\alpha_i+\beta_j$. Bayesian treatment of this model was studied e.g. by \cite{dealba:02}.

All the previous models assume independence in the data. To relax this assumption, 
several dependence models have been proposed. For continuous data, dependence across development years has been studied by \cite{kremer:05}, who proposed an autoregressive model of order one; and \cite{dealba&nieto:08} and \cite{nieto&targino:21}, who proposed gamma models with Markov and moving average dependencies, respectively. \cite{ntzoufras&dellaportas:02} proposed lognormal models and induce dependence across origin years via dynamic models. 

For discrete data, there have been fewer proposals. \cite{kremer:95} use an integer-valued autoregressive process (INAR) to induce dependence across development years; and \cite{bastos&al:19} use a negative binomial model with mean parameterisation with a logarithmic link and a linear predictor that accounts for row, column and row-column parameters, plus a linear dynamic prior specification. Our objective is to propose a parametric model for integer data that is able to accommodate dependencies across development years. We achieve this by considering the Poisson dependence sequences of \cite{nieto:22} and extend it to have the desired parameterisation for runoff triangles and to have negative binomial marginal distributions, which are more flexible than the Poisson. 

The contents of the rest of the paper is as follows: In Section \ref{sec:model} we start by defining an independent negative binomial model with the appropriate parameterisation, we then extend it to include moving average dependencies of order $q$, and study its prior properties. In Section \ref{sec:inference} we carry out a bayesian inference of the model and characterise the posterior distributions. We implement numerical analyses in Section \ref{sec:numerical}, which includes simulated and real data and comparisons with alternative models. We conclude with some remarks in Section \ref{sec:conclusion}. 

Before we proceed we introduce some notation: $\po(\mu)$ denotes a Poisson distribution with mean (rate) $\mu$; $\nbi(r,p)$ denotes a negative binomial distribution with number of failures $r$ and probability of success $p$; $\ga(\alpha,\beta)$ denotes a gamma distribution with shape parameter $\alpha$ and rate parameter $\beta$ with mean $\alpha/\beta$; In general, we will add an argument upfront to denote the corresponding density, e.g. $\po(x\mid\mu)$ denotes a Poisson density evaluated at $x$.

\section{Models}
\label{sec:model}

\subsection{Independence model}

Let $\{X_{i,j}\}$ be a set of random variables associated with the number of claims made at origin year $i$ and development year $j$, for $i,j=1,\ldots,n$. We assume that 
\begin{equation}
\label{eq:indmodel}
X_{i,j}\sim\nbi(\alpha_i,1/(1+\pi_j)), 
\end{equation}
with $\alpha_i\in\NB$ and $\pi_j\in\RB^+$, independently for all $i$ and $j$, such that $\E(X_{i,j})=\alpha_i\pi_j$ and $\V(X_{i,j})=\alpha_i\pi_j(1+\pi_j)$. One of the key aspects of the negative binomial model, that differs from the Poisson model, is its overdispersion property which can be seen from the fact that $\V(X_{i,j})=\E(X_{i,j})(1+\pi_j)$ so $\V(X_{i,j})>\E(X_{i,j})$ since $\pi_j>0$. 

The chosen parameterisation in the negative binomial model is convenient, however, as in most stochastic reserving models \citep[e.g.][]{dealba&nieto:08}, estimability constraints have to be imposed in the column parameters $\pi_j$. Typically, $\sum_{j=1}^n\pi_j=1$, so that the row parameter $\alpha_i$ can be interpreted as the ultimate total number of claims and $\pi_j$ is the proportion of the total number of claims due in the development year $j$. 

\subsection{Dependence model}

Recently, \cite{nieto:22} introduced a dependence sequence of Poisson random variables with the property that the marginal distributions are all invariant Poisson with the same parameters. Dependence is induced via a set of latent Poisson variables in a moving average fashion. Since a negative binomial distribution can be seen as a mixture of Poisson distributions, we can obtain a dependence sequence of negative binomial random variables via mixtures.

The construction is as follows. Let $\bX=\{X_{i,j}\}$ be the set of variables of interest and $\bY=\{Y_{i,j}\}$ and $\bZ=\{Z_{i,j}\}$ two sets of latent variables; then the model is constructed hierarchically as
\begin{align}
\label{eq:depmodel}
\nonumber
Z_{i,j}&\simind\ga(\alpha_i,1/\pi_j) \\
Y_{i,j}\mid Z_{i,j}&\simind\po(Z_{i,j}\gamma_{i,j}) \\
\nonumber
X_{i,j}-\sum_{l=0}^qY_{i,j-l}\mid\bY,\bZ&\simind\po\left(Z_{i,j}-\sum_{l=0}^q Z_{i,j-l}\gamma_{i,j-l}\right),
\end{align}
for $i,j=1,\ldots,n$. 
Here $\btheta=\{\balpha,\bpi,\bgamma\}$ is the set of parameters with: $\balpha=\{\alpha_i\}$ where $\alpha_i\in\NB$ for $i=1,\ldots,n$; $\bpi=\{\pi_j\}\in\Pi$ where $\Pi$ is such that $\pi_j\in[0,1]$ for $j=1,\ldots,n$ and $\sum_{j=1}^n\pi_j=1$; and $\bgamma=\{\gamma_{i,j}\}\in\Gamma$ where $\Gamma$ is such that $\gamma_{i,j}\in\RB^+$ for $i,j=1,\ldots,n$ and, conditionally on the $Z_{i,j}$'s, $Z_{i,j}-\sum_{l=0}^q Z_{i,j-l}\gamma_{i,j-l}\geq 0$. 

In model \eqref{eq:depmodel}, parameters $\gamma_{ij}$'s define the strength of dependence and $q\geq 0$ is the order of dependence across development years. We define $Y_{i,j}\equiv 0$ and $Z_{i,j}\equiv 0$ with probability one (w.p.1) for $j\leq 0$. The joint distribution of all variables involved in \eqref{eq:depmodel} can be easily computed via $f(\bx,\by,\bz)=f(\bx\mid\by,\bz)f(\by\mid\bz)f(\bz)$ and from here the joint (marginal) distribution of the variables of interest, $f(\bx)$, can be obtained via marginalisation. 
Figure \ref{fig:graphm} illustrates a graphical representation of the dependence model, where the dependence is of order $q=2$.

To study the properties of model \eqref{eq:depmodel} we recall two results.  
\begin{equation}
\label{eq:res1}
\mbox{If}\quad Y\sim\po(z\gamma) \quad\mbox{and}\quad X-y\mid Y=y\sim\po(z(1-\gamma)) \quad\Rightarrow\quad X\sim\po(z),\quad\mbox{and} 
\end{equation}
\begin{equation}
\label{eq:res2}
\mbox{if}\quad Z\sim\ga(\alpha,1/\pi)\quad\mbox{and}\quad X\mid Z\sim\po(z)\quad\Rightarrow\quad X\sim\nbi(\alpha,1/(1+\pi)). 
\end{equation}

\begin{proposition}
\label{prop:mod}
Let $\{X_{i,j}\}$ for $i,j=1,\ldots,n$ be a finite sequence whose probability law is described by equations \eqref{eq:depmodel}. Then, 
\begin{enumerate}
\item[(i)] The marginal distribution for each $X_{i,j}$ is $\nbi(\alpha_i,1/(1+\pi_j))$ for all $i,j$.
\item[(ii)] The autocorrelation between $X_{i,j}$ and $X_{i,j+k}$, for $1\leq k\leq q$ is 
$$\Cr(X_{i,j},X_{i,j+k})=\frac{\sum_{l=0}^{q-k}\pi_{j-l}\gamma_{i,j-l}}{\sqrt{\pi_j(1+\pi_j)\pi_{j+k}(1+\pi_{j+k})}}$$
and zero for $k>q$. 
\end{enumerate}
\end{proposition}
\begin{proof} \\
For $(i)$ we note that given $\bZ$ the $Y_{i,j}$'s are independent, so using the additivity property of independent Poisson random variables, 
\begin{equation}
\label{eq:posum}
\sum_{l=0}^q Y_{i,j-l}\,\Big{|}\,\bZ\sim\po\left(\sum_{l=0}^q Z_{i,j-l}\,\gamma_{i,j-l}\right).
\end{equation}
Now, considering the third level equation in \eqref{eq:depmodel} and using the result \eqref{eq:res1} we obtain $X_{i,j}\mid\bZ\sim\po(Z_{i,j})$. 
Finally, considering the first equation in \eqref{eq:depmodel} and result \eqref{eq:res2} we get $X_{i,j}\sim\nbi(\alpha_i,1/(1+\pi_j))$. \\ For $(ii)$ we rely on conditional independence properties and the iterative covariance formula. Conditioning on $\bY,\bZ$, then $$\Cv(X_{i,j},X_{i,j+k})=\E\{\Cv(X_{i,j},X_{i,j+k}\mid\bY,\bZ)\}+\Cv\{\E(X_{i,j}\mid\bY,\bZ),\E(X_{i,j+k}\mid\bY,\bZ)\}.$$ 
The first term becomes zero due to conditional independence. The second term is rewritten as $$\Cv\left(Z_{i,j}-\sum_{l=0}^q Z_{i,j-l}\gamma_{i,j-l}+\sum_{l=0}^q Y_{i,j-l}\,,\,Z_{i,j+k}-\sum_{l=0}^q Z_{i,j+k-l}\gamma_{i,j+k-l}+\sum_{l=0}^q Y_{i,j+k-l}\right).$$ Applying the iterative covariance formulae for a second time, conditioning on $\bZ$, we get $$\E\left\{\Cv\left(\sum_{l=0}^q Y_{i,j-l},\sum_{l=0}^q Y_{i,j+k-l}\,\Big{|}\,\bZ\right)\right\}+\Cv\left(Z_{i,j},Z_{i,j+k}\right),$$
where the first term is the result of removing the additive constants $Z_{i,j}$'s and the second term is the result of substituting $\E(Y_{i,j}\mid\bZ)=Z_{i,j}\gamma_{i,j}$. 
Since $Y_{i.j}$'s are conditionally independent given $\bZ$, the first term reduces to the expected value of the variance of the common elements, that is, $\E\left\{\V(\sum_{l=0}^{q-k}Y_{i,j-l}\mid\bZ)\right\}$ and the second term is zero for $k>0$. Now, using \eqref{eq:posum} to compute the conditional variance,
$$\Cv(X_{i,j},X_{i,j+k})=\E\left(\sum_{l=0}^{q-k}Z_{i,j-l}\gamma_{i,j-l}\right)=\alpha_i\sum_{l=0}^{q-k}\pi_{j-l}\gamma_{i,j-l}.$$
Finally computing the marginal variances we get $\V(X_{i,j})=\alpha_i\pi_j(1+\pi_j)$. Standardising the covariance we get the result. 
\end{proof}

\medskip
Proposition \ref{prop:mod} tells us two interesting things. The dependence model \eqref{eq:depmodel} reduces mar\-gi\-na\-lly to the negative binomial model \eqref{eq:indmodel}, but with dependence across development years. 
The autocorrelation expression between $\{X_{i,j}\}$ is a function of the column parameters $\bpi$ and $\bgamma$ parameters, therefore we refer to $\bgamma$ as the dependence parameters of the model. If $\gamma_{i,j}=0$ for all $i$ and $j$ then $Y_{i,j}=0$ with probability one (w.p.1), so regardless of the value of $q$, the $X_{i,j}$'s become independent. Moreover, for $q=0$ equations \eqref{eq:depmodel} reduce to $Z_{i,j}\sim\ga(\alpha_i,1/\pi_j)$,  $Y_{i,j}\mid\bZ\sim\po(Z_{i,j}\gamma_{i,j})$ and $X_{i,j}-Y_{i,j}\mid\bY,\bZ\sim\po(Z_{i,j}(1-\gamma_{i,j}))$ where, as in the general case, the marginal distributions are $X_{i,j}\sim\nbi(\alpha_i,1/(1+\pi_j))$ but with independence across the $X_{i,j}$'s, so the effect of $\gamma_{i,j}$ vanishes when $q=0$. In summary, the strength of the dependence is controlled by the lag $q$ and the dependence parameters $\bgamma$, larger / smaller $q$ and larger / smaller $\gamma_{i,j}$'s induce stronger / weaker dependence.

It is not easy to see why the autocorrelation, given in Proposition \ref{prop:mod}, is bounded by one since $\pi_j\in[0,1]$ and $\gamma_{i,j}\geq 0$. However, the rate parameter of the conditional distribution of  $X_{i,j}$ given $\bY$ and $\bZ$, third equation in \eqref{eq:depmodel}, is $Z_{i,j}-\sum_{l=0}^q Z_{i,j-l}\gamma_{i,j-l}$, which must be nonnegative. Therefore, if we take the expected value we get, $\alpha_i\left(\pi_j-\sum_{l=0}^q \pi_j\gamma_{i,j-l}\right)\geq 0$. Since $\alpha_i\in\NB$ then $\pi_j\geq\sum_{l=0}^q \pi_j\gamma_{i,j-l}$ and analogously $\pi_{j+k}\geq\sum_{l=0}^q \pi_{j+k}\gamma_{i,j+k-l}$. Now, observing that the numerator is smaller than any of the two previous sums, it is now clearer why $\Cor(X_{i,j},X_{i,j+k})\in[0,1]$. Having a positive dependence is useful in modeling trends across development years in the triangle. 

\section{Bayesian inference}
\label{sec:inference}

As mentioned in Section \ref{sec:intro}, the available data consists of a runoff triangle as in Table \ref{tab:triangle}. In notation, let $X_{i,j}$ be the set of observations for $j=1,\ldots,n-i+1$ and $i=1,\ldots,n$. For simplicity, we will make the dependence parameters independent of the row $i$, that is, $\gamma_{i,j}=\gamma_j$ to reduce the number of parameters in the model. This implies that the correlation $\rho_{j,j+k}=\Cor(X_{i,j},X_{i,j+k})$, given in Proposition \ref{prop:mod}, becomes independent of the origin year $i$. 

Here we describe a procedure, under a Bayesian approach, to make inferences about the unknown model parameters $\btheta = (\balpha, \bgamma, \bpi) $, where $\balpha=(\alpha_1,\ldots,\alpha_n)$, $\bgamma=(\gamma_1,\ldots,\gamma_n)$ and $\bpi=(\pi_1,\ldots,\pi_n)$. The posterior distribution will be characterised by the full conditional distributions of all elements in $\btheta$. 

To simplify the posterior derivation, we augment the likelihood by considering that the latent variables $\bZ$ and $\bY$ were also observed \citep[e.g.][]{tannerwong:87}. In the end, to obtain samples from the posterior distributions of model parameters, we will have to sample from the full conditional distributions of the latent variables $\bZ$ and $\bY$. 

The augmented likelihood function for $\btheta$ is given by 
\begin{align*}
f(\bx,\by,\bz \mid \btheta) &\propto \prod_{i=1}^n \prod_{j=1}^{n-i+1} \po\left(x_{ij}-\sum_{l=0}^q y_{i,j-l}\,\Bigg|\,z_{i,j}-\sum_{l=0}^q z_{i,j-l} \gamma_{j-l}\right) \\ 
&\hspace{2cm}\times \po(y_{i,j} \mid z_{i,j}\gamma_{j})\,\ga(z_{i,j} \mid \alpha_i,1/\pi_j).
\end{align*}

The prior distributions for each of the sets of parameters are assumed independent and are given by
$\alpha_i\sim\geo(p_{\alpha})$, for $i=1,\ldots,n$; $\gamma_j\sim\ga(a_{\gamma},b_{\gamma})$ for $j=1,\ldots,n$; and $\bpi\sim\dir(\ba)$, with $\ba=(a_1,\ldots,a_n)$ and $a_j>0$ for $j=1,\ldots,n$.

The full conditional distributions for the model parameters and the latent variables are given below. 
\begin{enumerate}

\item[(I)] Posterior conditional distribution for $\alpha_i$, $i=1,\ldots,n$
$$f(\alpha_i\mid\rest)\propto \frac{\left\{(1-p_{\alpha})\prod_{j=1}^{n-i+1}\left(z_{i,j}/\pi_j\right)\right\}^{\alpha_i}}{\{\Gamma(\alpha_i)\}^{n-i+1}}I_{\{0,1,\ldots\}}(\alpha_i)$$

\item[(II)] Posterior conditional distribution for $\gamma_j$, $j=1,\ldots,n$
$$\hspace{-2cm}f(\gamma_j\mid\rest)\propto\left\{\prod_{k=j}^{\min(j+q,n)}\prod_{i=1}^{n-k+1}\left(z_{i,k}-\sum_{l=0}^q z_{i,k-l}\gamma_{k-l}\right)^{x_{i,k}-\sum_{l=0}^q y_{i,k-l}}\right\}$$
$$\hspace{3cm}\times\gamma_j^{a_{\gamma}+\sum_{i=1}^{n-j+1}y_{i,j}-1}e^{-\gamma_j\left(b_{\gamma}+\sum_{i=1}^{n-j+1}z_{i,j}-\sum_{k=j}^{\min(j+q,n)}\sum_{i=1}^{n-k+1}z_{i,j}\right)},$$
for $0\leq\gamma_j\leq\min_{k=j,\ldots,\min(j+q,n);i=1,\ldots,n-k+1}\left\{\frac{z_{i,k}-\sum_{l=0,l\neq k-j}^q z_{i,k-l}\gamma_{k-l}}{z_{i,j}}\right\}$ if $q\geq 1$, and $0\leq\gamma_j\leq 1$ if $q=0$.

\item[(III)] Posterior conditional distribution for $\pi_j$, $j=1,\ldots,n-1$
$$f(\pi_j\mid\rest)\propto \pi_j^{a_j-1-\sum_{i=1}^{n-j+1}\alpha_i}\,e^{-(1/\pi_j)\sum_{i=1}^{n-j+1}z_{i,j}}\,\pi_n^{a_n-1-\alpha_1}\,e^{-z_{1,j}/\pi_n},$$
where $\pi_n=1-\sum_{k=1}^{n-1}\pi_k$, for $0\leq \pi_j\leq 1-\sum_{k=1,k\neq j}^{n-1}\pi_k$. 

\item[(IV)] Posterior distribution for $Y_{i,j}$, for $i=1,\ldots,n$, $j=1,\ldots,n-i+1$
$$\hspace{-7mm}f(y_{i,j}\mid\rest)\propto\left\{\frac{z_{i,j}\gamma_j}{\prod_{k=j}^{\min(j+q,n-i+1)}\left(z_{i,k}-\sum_{l=0}^q z_{i,k-l}\gamma_{k-l}\right)}\right\}^{y_{i,j}}$$ $$\hspace{1cm}\times\frac{1}{y_{i,j}!\prod_{k=j}^{\min(j+q,n-i+1)}\left(x_{i,k}-\sum_{l=0}^q y_{i,k-l}\right)!},$$
for $y_{i,j}\in\{0,\ldots,\min_{k=j\ldots,\min(j+q,n-i+1)}(x_{i,k}-\sum_{l=0,l\neq k-j}^q y_{i,k-l})\}$ if $q\geq 1$, and $y_{i,j}\in\{0,\ldots,x_{i,j}\}$ if $q=0$.

\item[(V)] Posterior conditional distribution for $Z_{i,j}$, for $i=1,\ldots,n$, $j=1,\ldots,n-i+1$
$$f(z_{i,j}\mid\rest)\propto \left\{\prod_{k=j}^{\min(j+q,n)}\left(z_{i,k}-\sum_{l=0}^q z_{i,k-l}\gamma_{k-l}\right)^{x_{i,k}-\sum_{l=0}^q y_{i,k-l}}\right\}$$
$$\hspace{1cm}\times\,z_{i,j}^{\alpha_i+y_{i,j}-1}e^{-z_{i,j}\{1/\pi_j-(\min(j+q,n-i+1)-j)\gamma_j+1\}},$$
for $z_{i,j}=0,1,\ldots$.
\end{enumerate}

With conditional distributions (I)--(V) we implement a Gibbs sam\-pler \citep{smith&roberts:93}. None of these distributions is of standard form; therefore, we require a Metropolis-Hastings step \citep{tierney:94} to sample from each of them. We define random walks with uniform proposal distributions. Specifically, for each parameter/latent variable $\theta\in\{\alpha_i,\gamma_j,\pi_j,y_{i,j},z_{i,j}\}$, at iteration $(t+1)$ we propose $\theta^*\mid\theta^{(t)}\sim\un(\theta^{(t)}-\delta_{\theta},\theta^{(t)}+\delta_{\theta})$ and take $\theta^{(t+1)}=\theta^*$ with probability $f(\theta^*\mid\rest)/f(\theta^{(t)}\mid\rest)$ and $\theta^{(t+1)}=\theta^{(t)}$ otherwise. Parameters $\delta_{\theta}$ are tuning parameters that determine the probability of acceptance for each node $\theta$. These were calibrated so that the acceptance probability is around $30\%$. 

We assess model fit by computing three statistics. the first one is the logarithm of the pseudo marginal likelihood (LPML), which is a measure of the predictive performance of the model. This is defined as a function of the Conditional Predictive Ordinate (CPO) $$LPML=\sum_{i=1}^n\sum_{j=1}^{n-i+1}\log(CPO_i),$$ with $CPO_i=f(x_{i,j}\mid \bx_{-(i,j)})$, and can be easily approximated via Monte Carlo, see \cite{geisser&eddy:79}. The second measure is the average squared bias defined as $$BIAS=\frac{2}{n(n+1)}\sum_{i=1}^n\sum_{j=1}^{n-i+1}\left\{\E(X_{i,j}\mid\data)-x_{i,j}\right\}^2,$$ and the third measure is the average predictive variance defined as $$PVAR=\frac{2}{n(n+1)}\sum_{i=1}^n\sum_{j=1}^{n-i+1}\V(X_{i,j}\mid\data).$$
Larger values of LPML and smaller values of BIAS and PVAR indicate a better fit. 

\section{Numerical analyses}
\label{sec:numerical}

\subsection{Simulation study}

We first test the posterior sampling algorithm in a control setting. We sample data from our dependence negative binomial model \eqref{eq:depmodel} with the following specifications: We take $n=10$ years to define the triangle; $\alpha_i=1000$ for $i=1,\ldots,n$ as the total number of claims; $\pi_j=2(n-j+1)/(n(n+1))$
for $j=1,\ldots,n$ as the development years proportions, which show a decreasing pattern, as in real life situations; $\gamma_j=0.15$ for $j=1,\ldots,n$ as the strength of dependence, with $q=2$ as order (lag) of dependence. 

Prior specifications for the model are: $p_{\alpha}=0.01$, $a_{\gamma}=1$, $b_{\gamma}=2$ and $a_j=1/2$ for $j=1,\ldots,n$. We ran the MCMC sampler for $50,000$ iterations, a burn-in of $5,000$ and a thinning of $20$. The code was implemented in Fortran in an Intel Xenon at 3.00 GHz with 24GB of RAM and a Linux operating system. Each run took 1.5 minutes. To determine the order of dependence, we took a set of different values $q=0,1,\ldots,4$ and compared each fit using the three statistics: LPML, BIAS and PVAR. Results are shown in Table \ref{tab:sim}. 
The three fit statistics select the true order of dependence, $q=2$, as the best fitting. 

Posterior estimates of the model parameters are included in Figures \ref{fig:alpha} and \ref{fig:gamma}, left panel. In all cases, 95\% posterior credible intervals (CI) contain the true value, with the intervals slightly larger for $\alpha_{10}$ and $\gamma_{10}$, because the posterior inference relies on only one observation (see Figures \ref{fig:alpha} and \ref{fig:gamma}, left panels). In Figure \ref{fig:alpha} (right panel) we observe the decreasing pattern of the development year proportions $\pi_j$, and our model is able to capture such behaviour. 

We test the prediction power of our model by producing 95\% posterior predictive CI. These are included in Figure \ref{fig:xf}. Observed data $x_{i,j}$ for $j=1,\ldots,n-i+1$ and $i=1,\ldots,n$ are denoted as dots, whereas CI are denoted as vertical lines. All observations are captured by our interval predictions. For the non-observed data, lower-right part of the triangle, $X_{i,j}$ for $j=n-i+2,\ldots,n$ and $i=2,\ldots,n$, we also produce 95\% CI and depict them as dotted lines. Our predictions follow the decreasing trend observed in the data, across development years. 

Finally, for each incomplete origin year, $i=2,\ldots,10$, we add the predicted unobserved number of claims $N_i=\sum_{j=n-i+2}^n X_{i,j}$ and report them as boxplots in Figure \ref{fig:gamma} (right panel). As expected, the unobserved number of claims is increasing as we increase the origin year. For the last two years $i=9,10$, the predicted number of claims is around the same quantity, with more dispersion shown for year $i=10$. Overall, adding the total number of future claims $N=\sum_{i=2}^{10}N_i$, the model predicts an average of $2,858$ claims with a $95\%$ CI between $2,605$ and $3,129$ claims. 

\subsection{General insurance data}

One of the most analysed datasets of number of claims is that of a portfolio of general insurance policies, which is reported, for example, in \cite{dealba:02} and also included in Table \ref{tab:dealba}. The information is available for $n=10$ years. 

We fit our negative binomial dependence model with the same prior specifications as in the simulation study. MCMC sampler was run for 50,000 iterations with a burn in of 5,000 and keeping one of every 20$^{th}$ iteration to compute posterior summaries. The running time was 1.5 minutes. 

To select the order of dependence, we fitted the model with varying $q\in\{0,1,2,3\}$ and compared the fitting statistics. These are reported in Table \ref{tab:dealbagof}. According to LPML and PVAR, the best model is obtained when $q=1$, however, the BIAS selects the model that assumes independence, which is obtained when $q=0$. We therefore compare posterior summaries for both models. 

Figures \ref{fig:alphad} and \ref{fig:gammad} present posterior estimates obtained with $q=1$. Point and 95\% CI, for parameters $\alpha_i$, $i=1,\ldots,n$ are included in the left panel of Figure \ref{fig:alphad}. These show an interesting pattern in the ultimate total number of claims, it starts around 600 in the first year, then increases up to 700 in year two and from there it steadily decreases until around 300 claims in year ten. 

The right panel in Figure \ref{fig:alphad} presents posterior estimates for development year proportions $\pi_j$ for $j=1,\ldots,n$. The typical pattern of development year proportions is decreasing in time; however, for these data, the pattern is of mountain shape with an increasing tendency for years one to three and a decreasing tendency for years from fourth to tenth. We also note that the uncertainty for years from six to ten is very low. 

In Figure \ref{fig:gammad}, left panel, we present point and 95\% CI estimates for the dependence parameters $\gamma_j$, $j=1,\ldots,n$. Recall that these parameters determine the strength of the dependence across development years. For the first year, there is a lot of uncertainty in $\gamma_1$ with values close to 0.4, however, for years two to five the uncertainty is highly reduced with values around 0.1. The uncertainty starts to increase for years beyond five, perhaps due to the reduced number of observations and with values around 0.2. 

In the right panel of Figure \ref{fig:gammad}, we include posterior estimates of the correlations given in Proposition \ref{prop:mod}. Since the order of dependence is $q=1$, the correlations are only positive between adjacent years $j$ and $j+1$, so we denote them as $\rho_{j,j+1}$. Although the pattern is similar to the dependence parameters $\gamma_j$, the correlations have been standardised with the development years proportions and are bounded to a $[0,1]$ scale. 

We finally compute the predictive distribution of the aggregated number of unobserved claims $N_i$ for each origin year $i=2,\ldots,n$. To place our predictions in context, we also computed the chain-ladder predictions and included them as bold numbers in Table \ref{tab:dealba} (lower-down triangle). The predictive distributions with both models, independent ($q=0$) and dependent ($q=1$) are presented in Figure \ref{fig:resd} as grey boxplots. Chain-ladder estimates are indicated as red asteriks. For the independence model, left panel, chain-ladder estimates lie inside each of the boxes that represents 50\% probability, however, for the dependence model, right panel, the chain-ladder estimate for year ten lies slightly outside of the box, but certainly within the 95\% probability interval. 

The total aggregated number of claims $N$, for years $j=2,\ldots,n$, is presented in Figure \ref{fig:restotd}. For the independence model (left panel), the chain-ladder estimate (vertical line) lies in the center of the predictive distribution, whereas for the dependence model (right panel), the chain-ladder estimate lies towards the right tail of the distribution, which indicates that the chain-ladder estimates of $902$ is very conservative, if considering the posterior predictive mean with the dependence model, which is $818$ claims. 

\subsection{Automobile data}

We now consider a data set on automobile bodily injury liability \citep{berquist&sherman:77}. This consists of claim counts for the period from 1969 to 1976, i.e., for a total of $n=8$ years. The data is available in Table \ref{tab:auto}. We note that the numbers for the first two development years are a lot bigger than the number for the development years from three to eight, this means that the great majority of claims occur in the first two development years. 

We fit our negative binomial dependence model with the same prior specifications as in the simulation study, but this time we run longer chains. MCMC was run for 100,000 iterations with a burn in of 10,000 and a thinning of 40. The running time was 1.3 minutes. 

Since the runoff triangle is of dimension $n=8$, we selected the order of dependence by considering the values $q\in\{0,1,2\}$. The corresponding fit statistics are reported in Table \ref{tab:autogof}. The three statistics LPML, BIAS and PVAR all select the model with $q=1$ to be the best. We therefore report inferences with this model. 

The posterior estimates for the ultimate number of claims, $\alpha_i$ for $i=1,\ldots,n$, are reported in Figure \ref{fig:alphaa} (left panel). The trend across the different origin years is not smooth. It starts around 7,800 claims in the first year and linearly increases up to 9,500 in year three and remains there for the next two years. It comes down to 7,800 in year six, goes a little up to 8,000 in year seven, and finally comes down to 7,200 in year eight. 

Estimates of the proportion of claims that occurred in each development year, $\pi_j$ for $j=1,\ldots,n$, are included in the right panel of Figure \ref{fig:alphaa}. As already seen in the data, the first year represents a little more than 80\% of the claims, and year two a little less than 20\% of the claims. The proportion of claims due in years three to eight is close to zero. 

The posterior estimates for the dependence parameters, $\gamma_j$ for $j=1,\ldots,n$, are reported in the left panel in Figure \ref{fig:gammaa}. We note an increase in uncertainty as the development years evolve. This is reflected in wider credible intervals. To better appreciate the dependence between years, we computed the correlations $\rho_{j,j+1}$ for $i=1,\ldots,7$. Considering the point estimates, we note that the correlation fluctuates around 0.2, being a little lower for the first two years. The size of the intervals in the first two correlations with respect to the others is perhaps due to the magnitude of the numbers. 

The aggregated number of claims, $N_i$ for $i=2,\ldots,n$, are also estimated. These are included in the left panel of Figure \ref{fig:resa}. The posterior predictive distributions for $N_i$ are shown as gray box plots, and the chain-ladder estimates are denoted by red asterisks. Apart from the last two years, the chain-ladder estimates lie inside the 50\% middle box. For year seven, our 95\% CI is $N_7\in[113,165]$ and the chain-ladder estimate is 160, which is still inside. However, for year eight, $N_8\in[1079,1242]$ and the chain-ladder estimate is 1343, which is clearly outside of our prediction interval. 

The posterior predictive distribution for the total aggregated number of claims, $N$, is included as a histogram in the right panel of Figure \ref{fig:resa}. The chain-ladder point estimate is in the limit of the right tail of our predictive distribution. In fact, the posterior predictive mean is 1397 with a 95\% CI of $[1309,1484]$, while the chain-ladder value is 1597, which is clearly not supported by our model and overestimates the number of claims. 

\section{Concluding remarks}
\label{sec:conclusion}

We have introduced a negative binomial model with an appealing parameterisation in terms of row and column parameters. We have also extended the model to include dependence or order $q\geq 0$ as in a moving average fashion, but maintaning the marginal distribution as in the independence case. 

Posterior inference of our model requires the implementation of an MCMC algorithm with five sets of conditional distributions, three sets of parameters plus two sets of latent variables. The algorithm was implemented in Fortran, which makes it very efficient. The executable files can run through R without the need for additional compilation, and the code is available as supplementary material. 

For the data sets analyses here, we have shown that there is dependence across development years, and ignoring it results in overestimating the reserve, which is a waste of resources for insurance companies. 

Possible extensions of our model are the inclusion of dependence across origin years and/or calendar years (diagonals in the runoff triangle). These and other possible extensions are left for future work. 

\section*{Acknowledgements}

The first author acknowledges support from \textit{Asociaci\'on Mexicana de Cultura, A.C.}. The second author acknowledges support from CNPq (200293/2022-2) and FAPERJ (E-26/201.350, E-26/211.426, E-26/211.578).

\bibliographystyle{natbib}

\newpage


\begin{table}
$$\begin{array}{c|ccccccc} 
\hline \hline
\text{Year of} & \multicolumn{7}{c}{\text{Development year}} \\ 
\text{origin} & 1 & 2 & \cdots & j & \cdots & n-1 & n \\ \hline
1 & X_{1,1} & X_{1,2} & \cdots & X_{1,j} &  & X_{1,n-1} & X_{1,n} \\ 
2 & X_{2,1} & X_{2,2} & \cdots & X_{2,j} &  & X_{2,n-1} &  \\ 
\vdots & \vdots & \vdots & \cdots & \vdots &  &  &  \\ 
i & X_{i,1} & X_{i,2} & \cdots & X_{i,n+1-i} &  &  &  \\ 
\vdots & \vdots & \vdots &  &  &  &  &  \\ 
n-1 & X_{n-1,1} & X_{n-1,2} &  &  &  &  &  \\ 
n & X_{n,1} &  &  &  &  &  &  \\ \hline \hline
\end{array}$$
\caption{Run-off triangle of available data.}
\label{tab:triangle}
\end{table}

\begin{table}
\centering
$$\begin{array}{cccc} \hline\hline
q & \text{LPML} & \text{BIAS} & \text{PVAR} \\ \hline
0 & -221.43 & 28.67 & 130.97 \\
1 & -219.40 & 26.38 & 106.40 \\
2 & \bf{-218.11} & \bf{25.94} & \bf{102.94} \\
3 & -220.09 & 31.57 & 104.47 \\
4 & -221.44 & 35.66 & 121.09 \\
\hline\hline
\end{array}$$
\caption{Fit statistics in simulation study. Best fitting in bold.}
\label{tab:sim}
\end{table}

\begin{table}
$$\begin{array}{c|cccccccccc} 
\hline \hline
i/j & 1 & 2 & 3 & 4 & 5 & 6 & 7 & 8 & 9 & 10 \\ \hline
1 & 40 & 124 & 157 & 93 & 141 & 22 & 14 & 10 & 3 & 2 \\ 
2 & 37 & 186 & 130 & 239 & 61 & 26 & 23 & 6 & 6 & \bf{2} \\ 
3 & 35 & 158 & 243 & 153 & 48 & 26 & 14 & 5 & \bf{5} & \bf{2} \\ 
4 & 41 & 155 & 218 & 100 & 67 & 17 & 6 & \bf{6} & \bf{4} & \bf{2} \\ 
5 & 30 & 187 & 166 & 120 & 55 & 13 & \bf{13} & \bf{6} & \bf{4} & \bf{2} \\ 
6 & 33 & 121 & 204 & 87 & 37 & \bf{17} & \bf{11} & \bf{5} & \bf{4} & \bf{2} \\ 
7 & 32 & 115 & 146 & 103 & \bf{53} & \bf{16} & \bf{11} & \bf{5} & \bf{3} & \bf{2} \\ 
8 & 43 & 111 & 83 & \bf{83} & \bf{43} & \bf{13} & \bf{9} & \bf{4} & \bf{3} & \bf{1} \\ 
9 & 17 & 92 & \bf{101} & \bf{74} & \bf{38} & \bf{11} & \bf{8} & \bf{4} & \bf{2} & \bf{1} \\ 
10 & 22 & \bf{89} & \bf{103} & \bf{75} & \bf{39} & \bf{11} & \bf{8} & \bf{4} & \bf{2} & \bf{1} \\ 
\hline \hline
\end{array}$$
\caption{General insurance data. Observed data (upper-left triangle) and chain-ladder forecasts (bottom-right triangle).}
\label{tab:dealba}
\end{table}

\begin{table}
\centering
$$\begin{array}{cccc} \hline\hline
q & \text{LPML} & \text{BIAS} & \text{PVAR} \\ \hline
0 & -353 & \bf{200} & 108 \\
1 & \bf{-334} & 373 & \bf{102} \\
2 & -350 & 383 & 105 \\
3 & -354 & 398 & 115 \\
\hline\hline
\end{array}$$
\caption{Fit statistics in general insurance data. Best fitting in bold.}
\label{tab:dealbagof}
\end{table}

\begin{table}
$$\begin{array}{c|cccccccc} 
\hline \hline
i/j & 1 & 2 & 3 & 4 & 5 & 6 & 7 & 8 \\ \hline
1 & 6553 & 1143 & 74 & 29 & 15 & 5 & 1 & 1 \\
2 & 7277 & 1260 & 78 & 46 & 14 & 4 & 3 \\
3 & 8259 & 1506 & 119 & 42 & 14 & 5 \\
4 & 7858 & 1616 & 141 & 49 & 16	\\
5 & 7808 & 1568 & 137 & 49 & \\
6 & 6278 & 1336 & 127 & \\
7 & 6446 & 1438 \\
8 & 6115 \\
\hline \hline
\end{array}$$
\caption{Automobile observed data.}
\label{tab:auto}
\end{table}

\begin{table}
\centering
$$\begin{array}{cccc} \hline\hline
q & \text{LPML} & \text{BIAS} & \text{PVAR} \\ \hline
0 & -262 & 5240 & 3075 \\
1 & \bf{-225} & \bf{4526} & \bf{3022} \\
2 & -232 & 5061 & 3241 \\
\hline\hline
\end{array}$$
\caption{Fit statistics automobile data. Best fitting in bold.}
\label{tab:autogof}
\end{table}


\begin{figure}[h]
\setlength{\unitlength}{0.8cm}
\hspace{-5mm}
\begin{picture}(20,4.0)
\put(4.9,3.3){\oval(2.4,1.1)}
\put(4.0,3.1){$Z_1\rightarrow Y_1$}
\put(7.9,3.3){\oval(2.4,1.1)}
\put(7.0,3.1){$Z_2\rightarrow Y_2$}
\put(10.9,3.3){\oval(2.4,1.1)}
\put(10.0,3.1){$Z_3\rightarrow Y_3$}
\put(13.9,3.3){\oval(2.4,1.1)}
\put(13.0,3.1){$Z_4\rightarrow Y_4$}
\put(16.9,3.3){\oval(2.4,1.1)}
\put(16.0,3.1){$Z_5\rightarrow Y_5$}
\put(4.9,0.3){\oval(1.1,1.1)}
\put(4.6,0.1){$X_1$}
\put(7.9,0.3){\oval(1.1,1.1)}
\put(7.6,0.1){$X_2$}
\put(10.9,0.3){\oval(1.1,1.1)}
\put(10.6,0.1){$X_3$}
\put(13.9,0.3){\oval(1.1,1.1)}
\put(13.6,0.1){$X_4$}
\put(16.9,0.3){\oval(1.1,1.1)}
\put(16.6,0.1){$X_5$}
\put(4.8,2.7){\vector(0,-1){1.8}}
\put(7.8,2.7){\vector(0,-1){1.8}}
\put(10.8,2.7){\vector(0,-1){1.8}}
\put(13.8,2.7){\vector(0,-1){1.8}}
\put(16.8,2.7){\vector(0,-1){1.8}}
\put(4.8,2.7){\vector(4,-3){2.6}}
\put(7.8,2.7){\vector(4,-3){2.6}}
\put(10.8,2.7){\vector(4,-3){2.6}}
\put(13.8,2.7){\vector(4,-3){2.6}}
\put(4.8,2.7){\vector(3,-1){5.7}}
\put(7.8,2.7){\vector(3,-1){5.7}}
\put(10.8,2.7){\vector(3,-1){5.7}}
\end{picture}
\vspace{5mm}
\caption{Graphical representation of dependence model \eqref{eq:depmodel} for $q=2$.}
\label{fig:graphm} 
\end{figure}
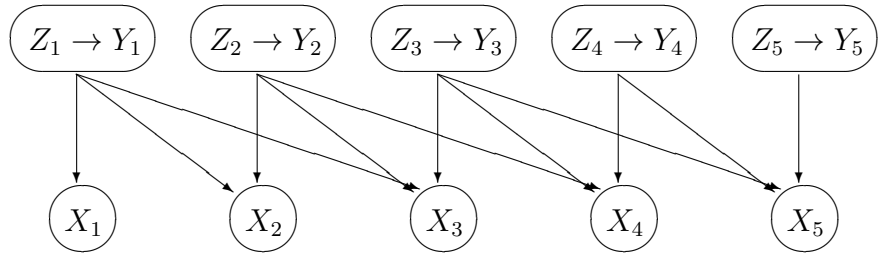

\begin{figure}
\centering
\includegraphics[scale=0.45]{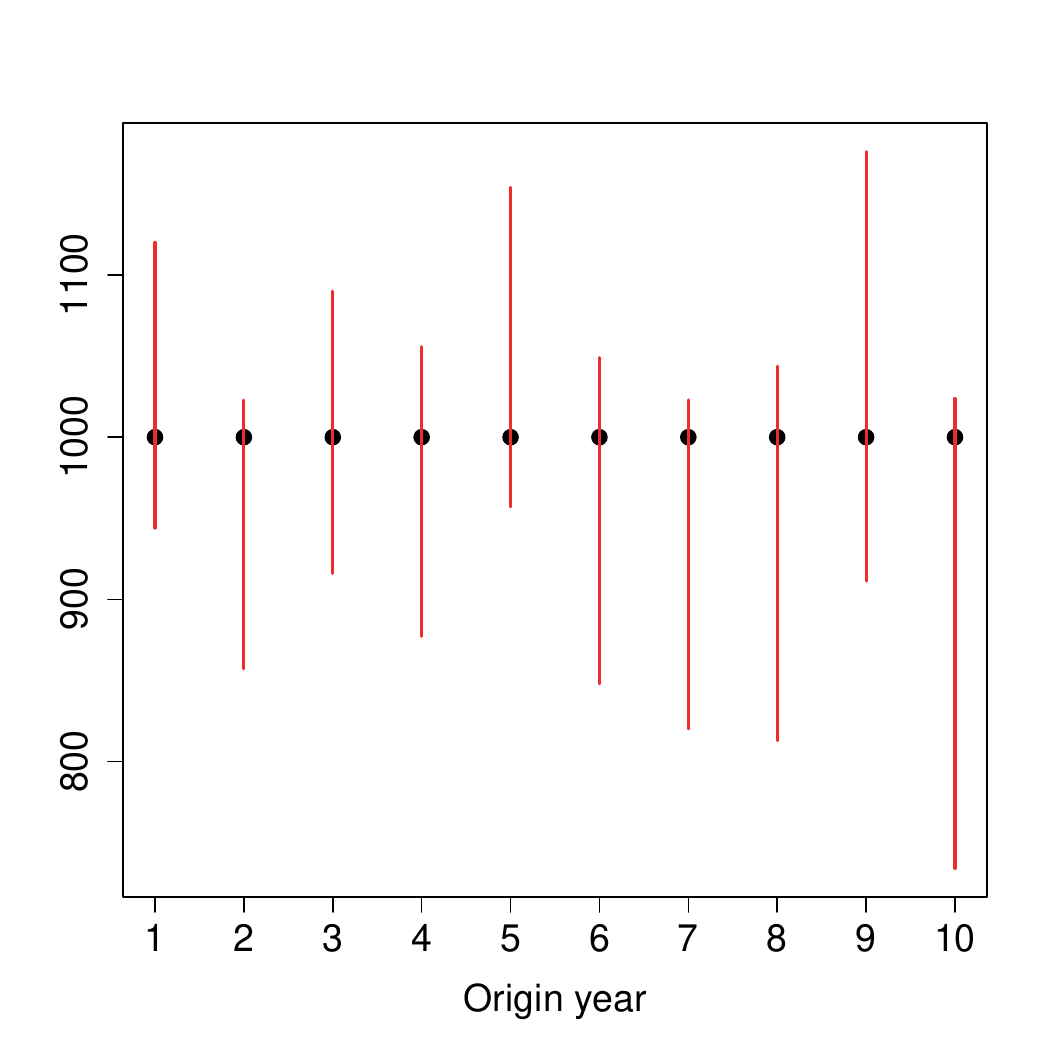}
\includegraphics[scale=0.45]{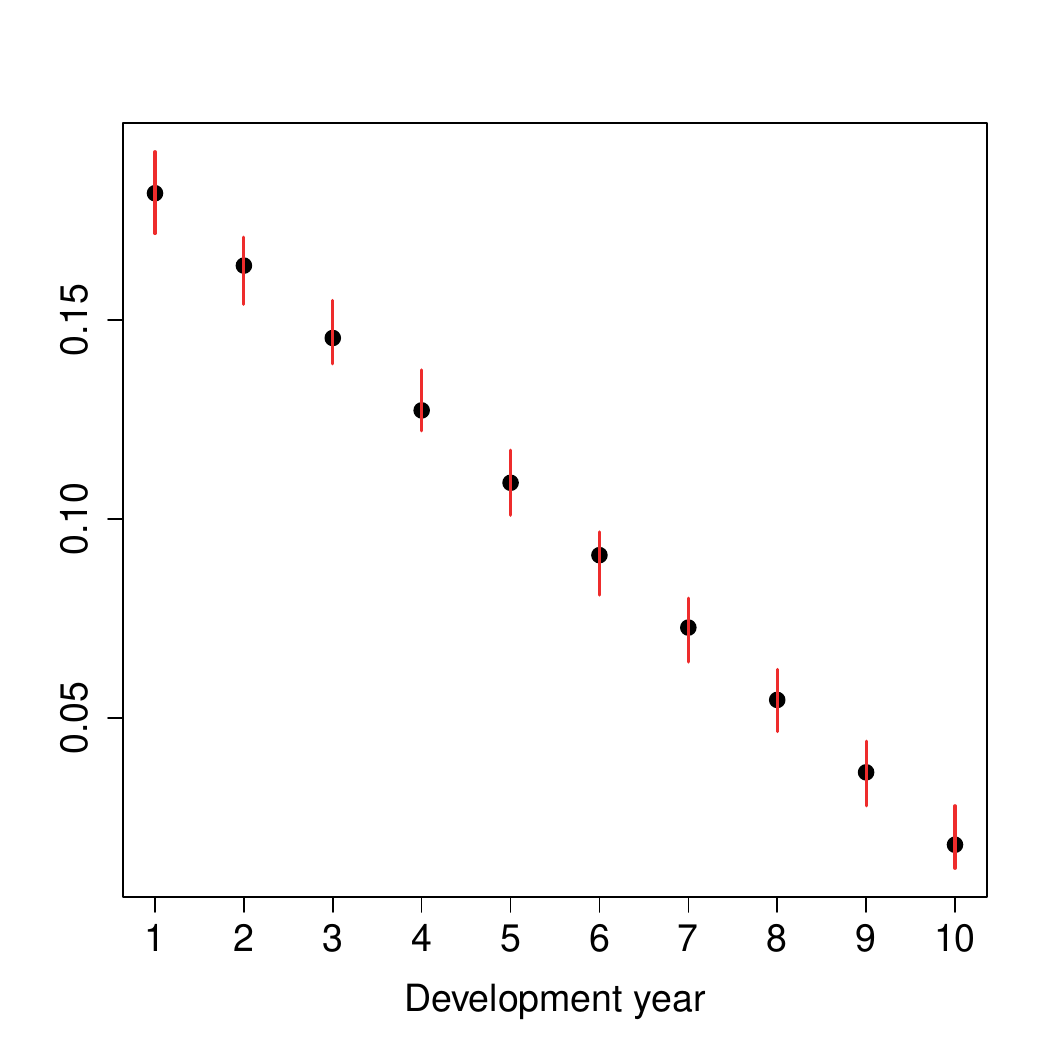}
\caption{Simulated data. Posterior estimates of parameters: $\alpha_i$, $i=1,\ldots,n$ (left) and $\pi_j$, $j=1,\ldots,n$ (right) with $n=10$. True value (dots) and 95\% CI (lines).}
\label{fig:alpha}
\end{figure}

\begin{figure}
\centering
\includegraphics[scale=0.45]{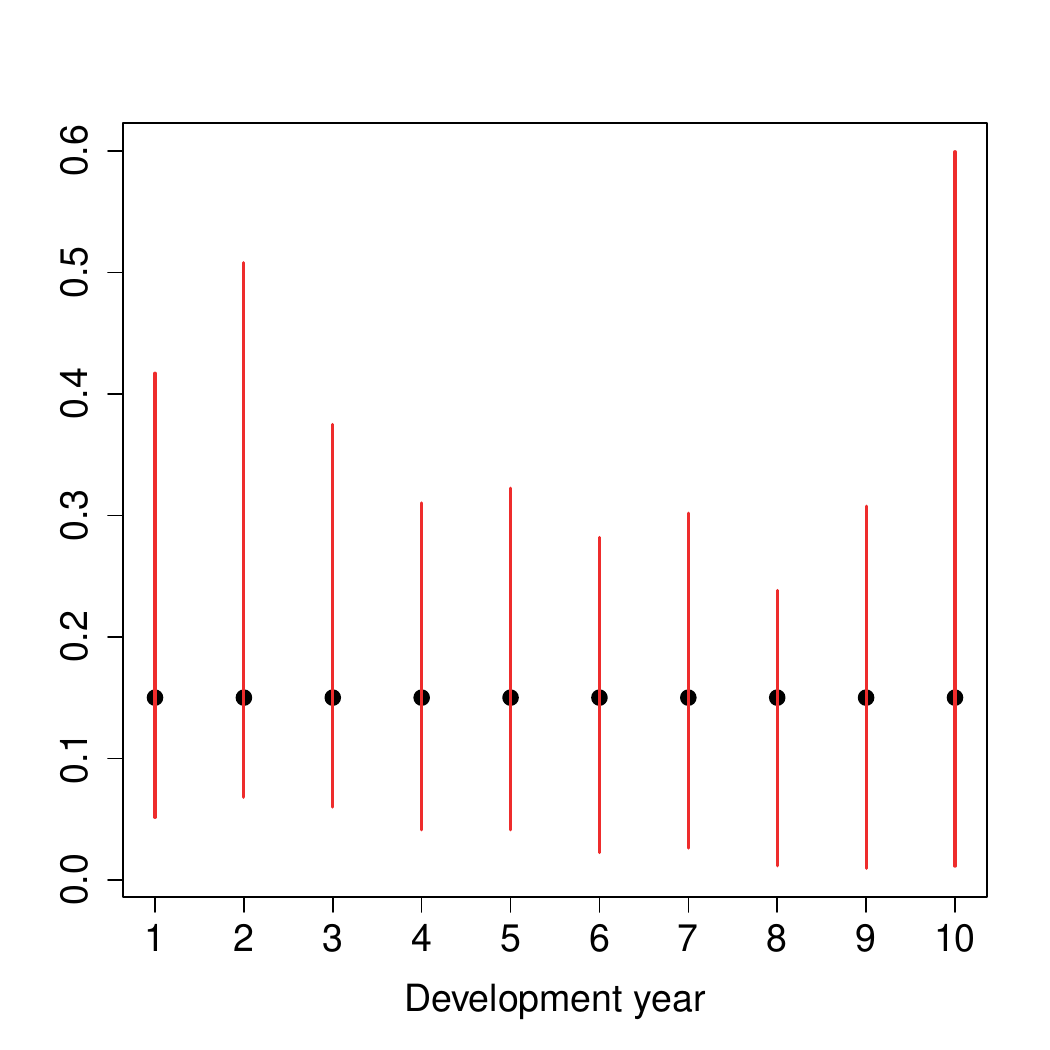}
\includegraphics[scale=0.45]{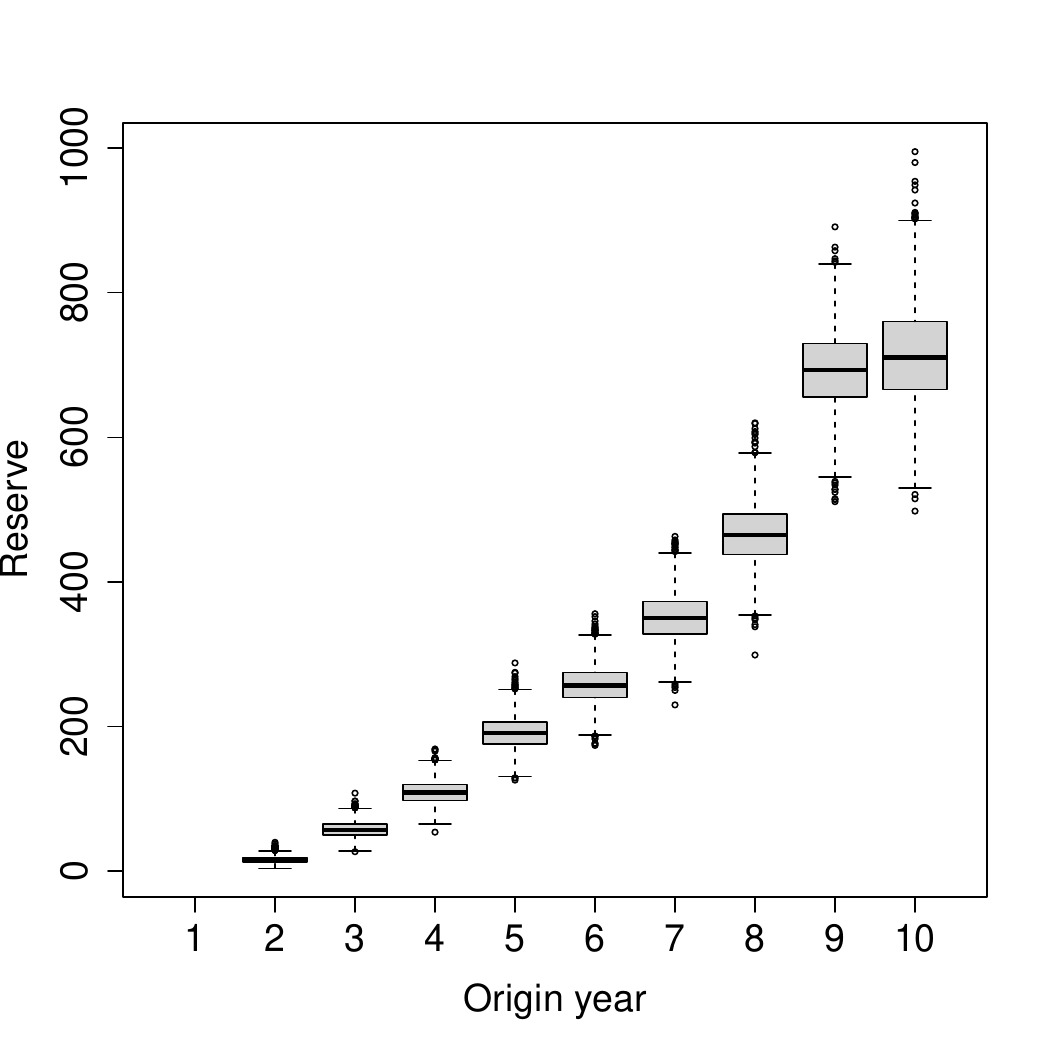}
\caption{Simulated data. Left: posterior estimates for parameters $\gamma_j$, $j=1,\ldots,n$ with $n=10$. True value (dots) and 95\% CI (lines). Right: boxplots of posterior predicted aggregated number of claims $N_i$, for $i=2,\ldots,n$.}
\label{fig:gamma}
\end{figure}

\begin{figure}
\centering
\includegraphics[scale=0.8]{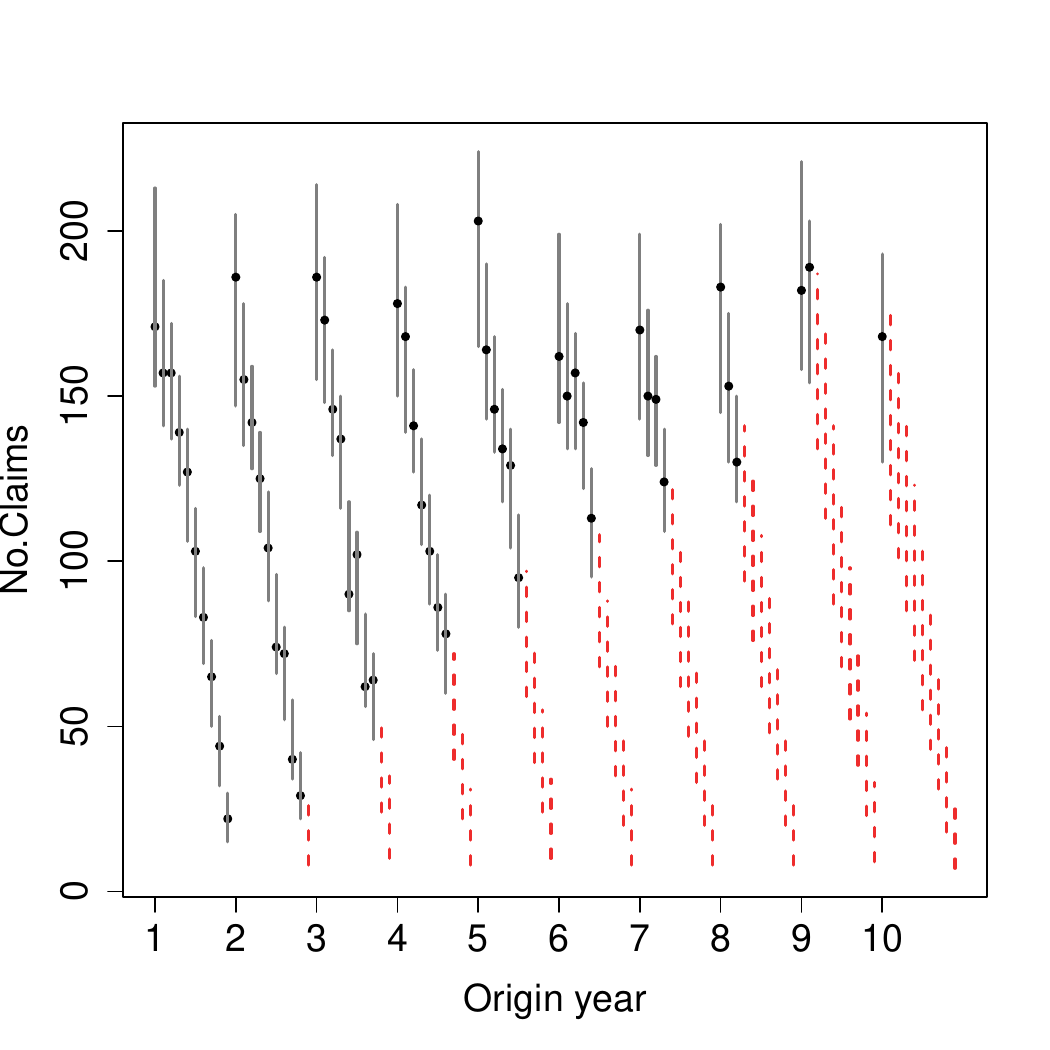}
\caption{Simulated data. Posterior predictions for $X_{i,j}$, $i,j=1,\ldots,n$ with $n=10$. True value (dots) and 95\% CI (lines). Within sample (solid lines) and out of sample (dotted lines).}
\label{fig:xf}
\end{figure}

\begin{figure}
\centering
\includegraphics[scale=0.45]{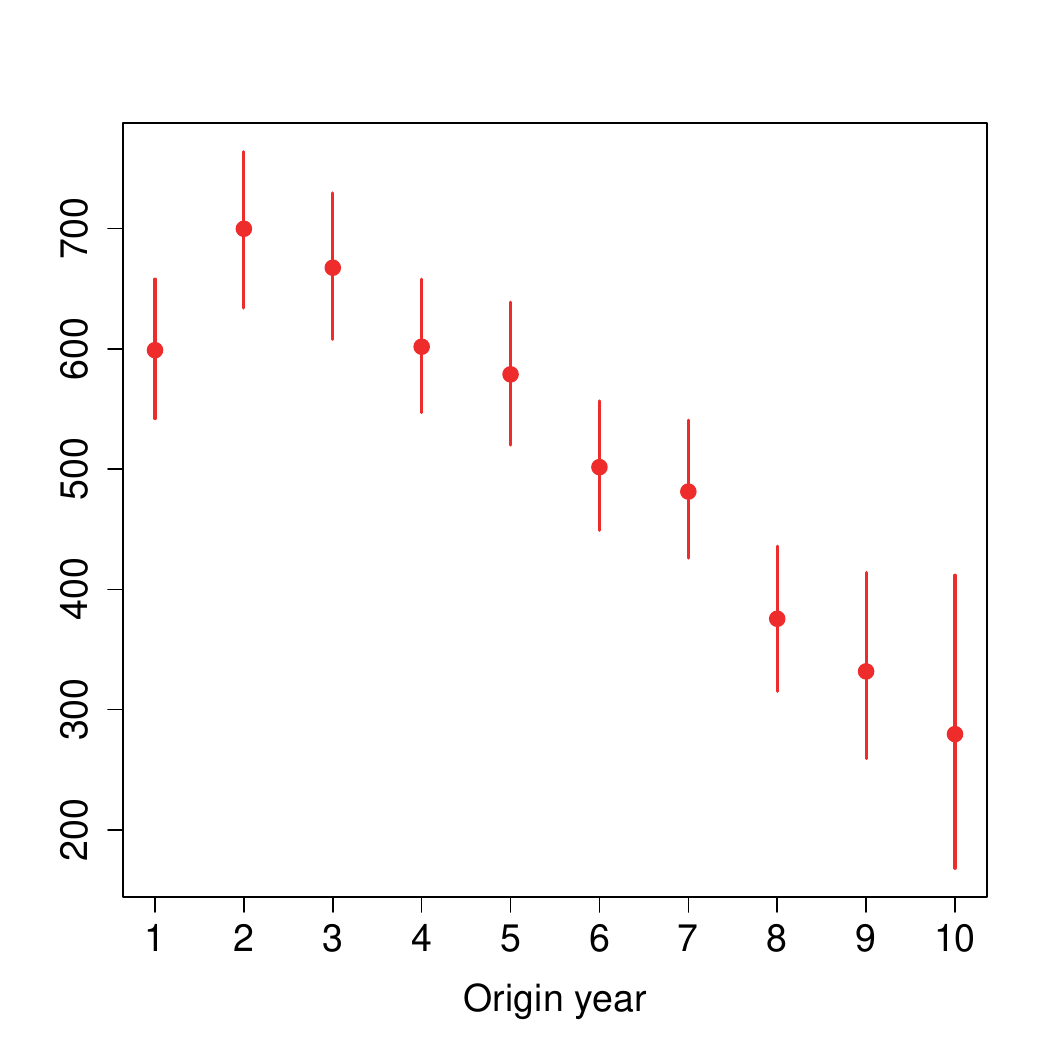}
\includegraphics[scale=0.45]{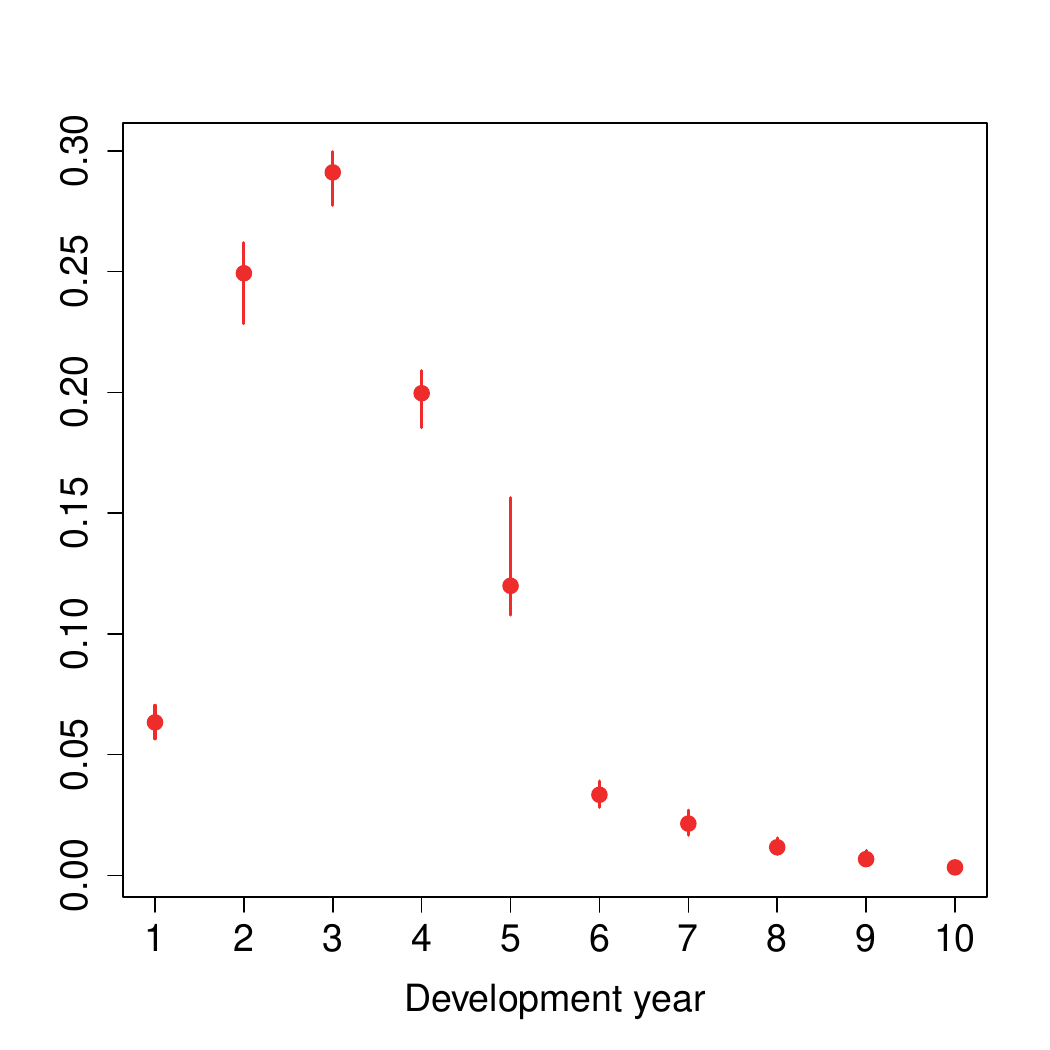}
\caption{General insurance data. Posterior estimates of parameters $\alpha_i$, $i=1,\ldots,n$ (left) and $\pi_j$, $j=1,\ldots,n$ (right) with $n=10$. Posterior mean (dots) and 95\% CI (lines).}
\label{fig:alphad}
\end{figure}

\begin{figure}
\centering
\includegraphics[scale=0.45]{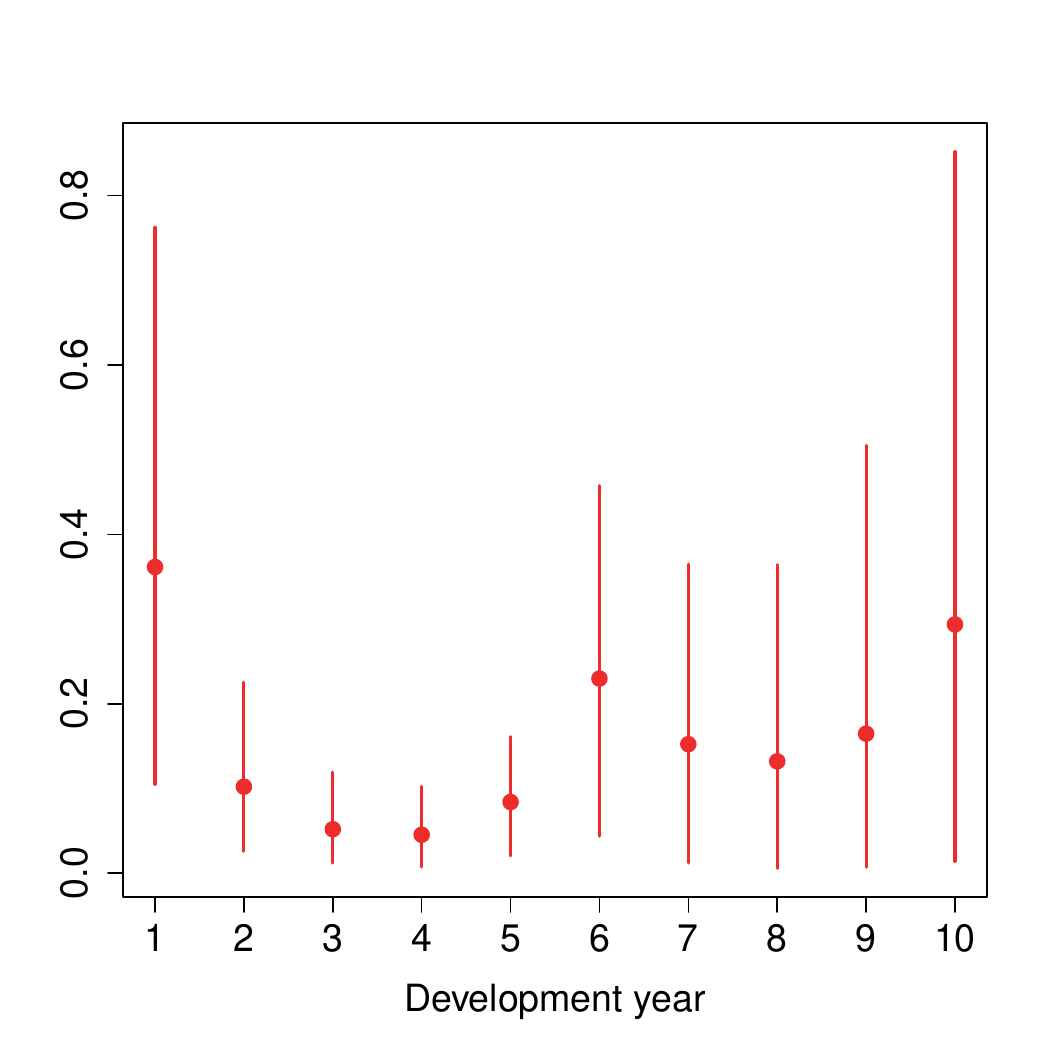}
\includegraphics[scale=0.45]{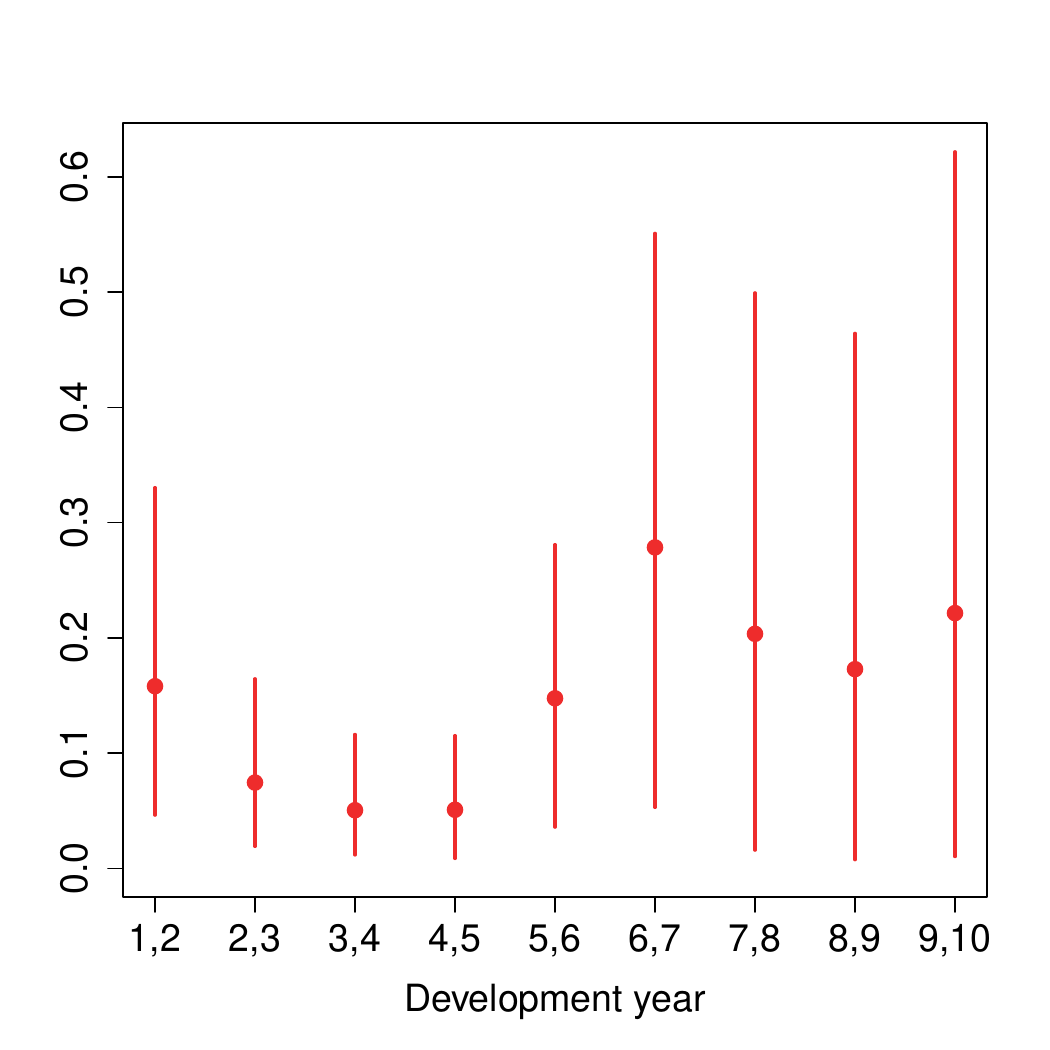}
\caption{General insurance data. Posterior estimates of parameters $\gamma_j$ (left) and $\rho_{j,j+1}$ (right) for $j=1,\ldots,n$ with $n=10$. Posterior mean (dots) and 95\% CI (lines).}
\label{fig:gammad}
\end{figure}

\begin{figure}
\centering
\includegraphics[scale=0.45]{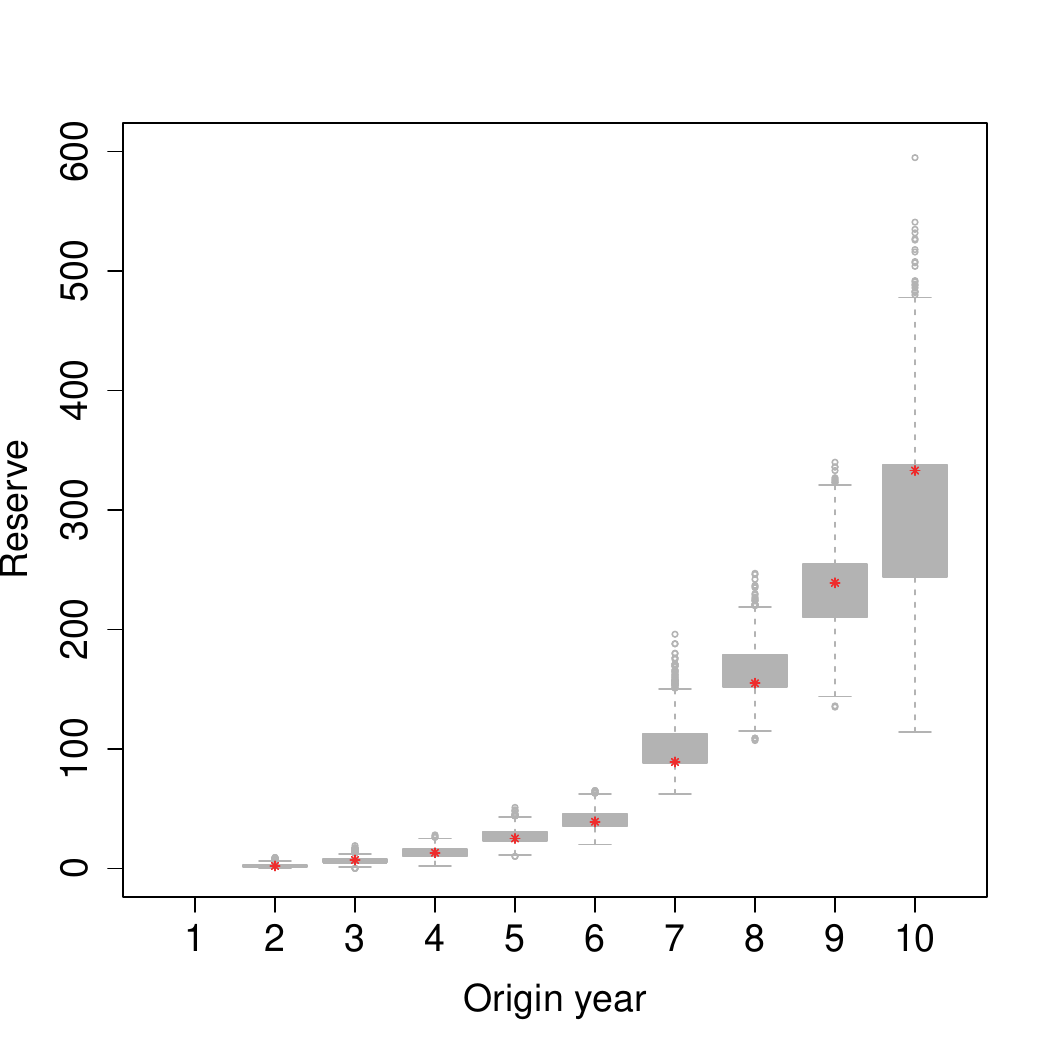}
\includegraphics[scale=0.45]{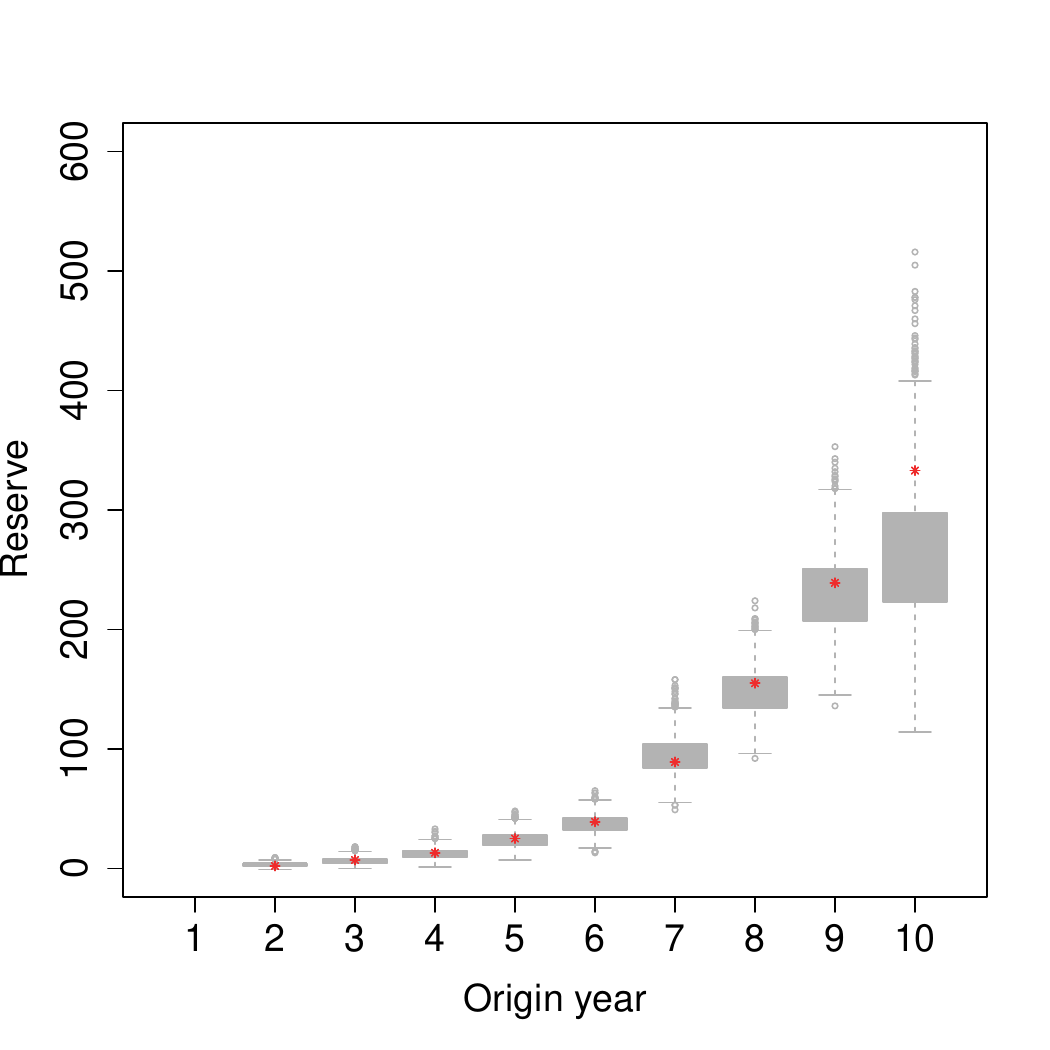}
\caption{General insurance data. Posterior predictive distributions of aggregated number of claims $N_i$, $i=1,\ldots,n$ with $n=10$. Independence model $q=0$ (left) and dependence model $q=1$ (right). Boxplots (grey) and chain ladder point estimates (asteriks).}
\label{fig:resd}
\end{figure}

\begin{figure}
\centering
\includegraphics[scale=0.45]{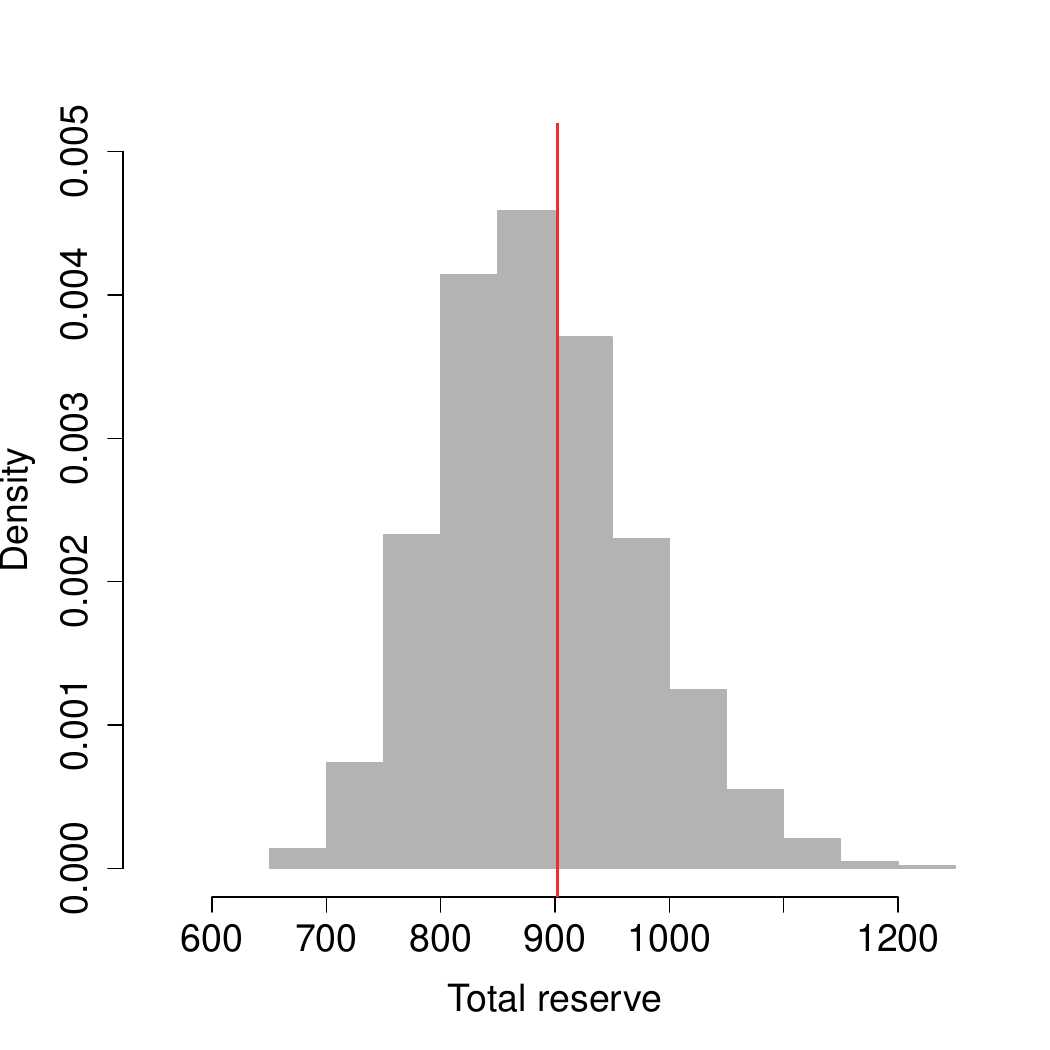}
\includegraphics[scale=0.45]{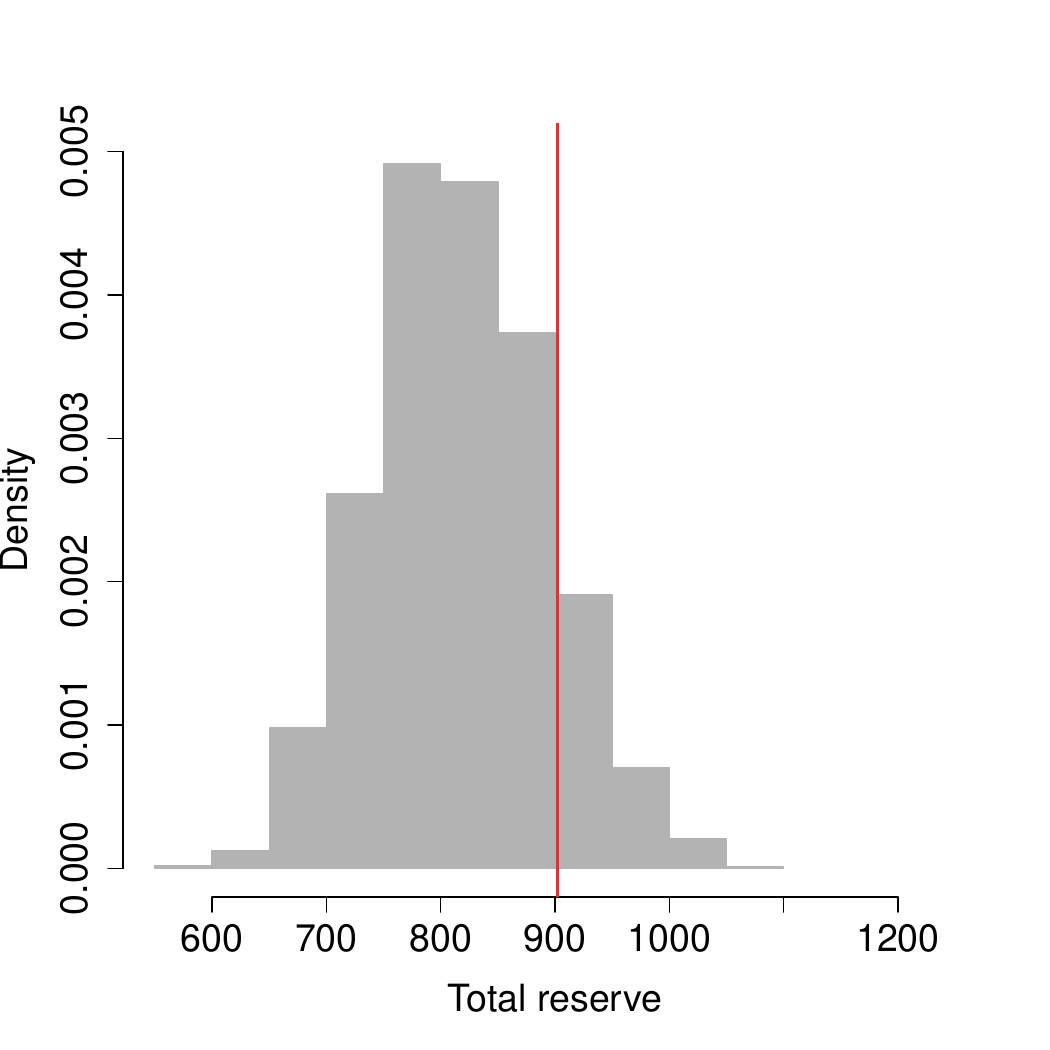}
\caption{General insurance data. Posterior predictive distributions of total aggregated number of claims $N$. Independence model $q=0$ (left) and dependence model $q=1$ (right). Histogram (grey) and chain ladder point estimates (asteriks).}
\label{fig:restotd}
\end{figure}

\begin{figure}
\centering
\includegraphics[scale=0.45]{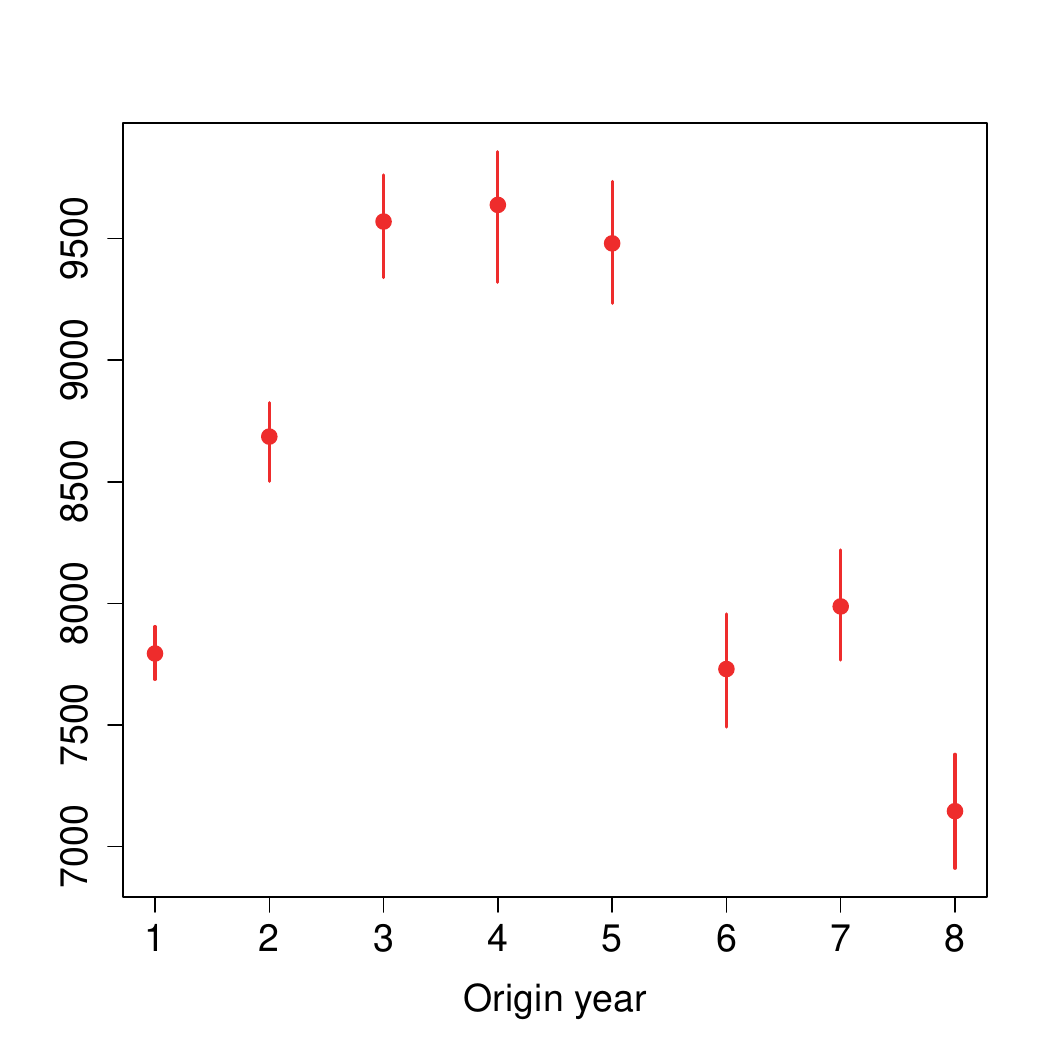}
\includegraphics[scale=0.45]{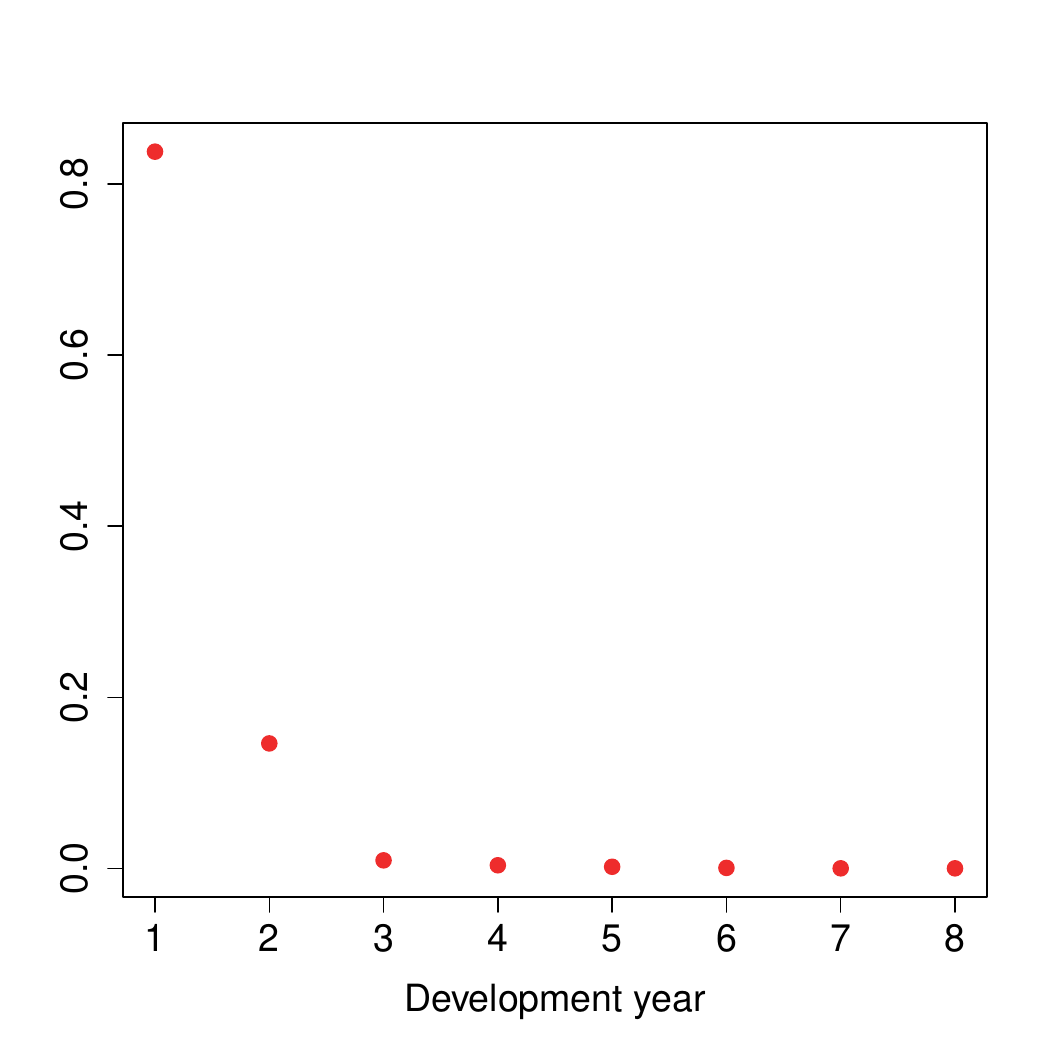}
\caption{Automobile data. Posterior estimates of parameters: $\alpha_i$, $i=1,\ldots,n$ (left) and $\pi_j$, $j=1,\ldots,n$ (right) with $n=8$. Posterior mean (dots) and 95\% CI (lines).}
\label{fig:alphaa}
\end{figure}

\begin{figure}
\centering
\includegraphics[scale=0.45]{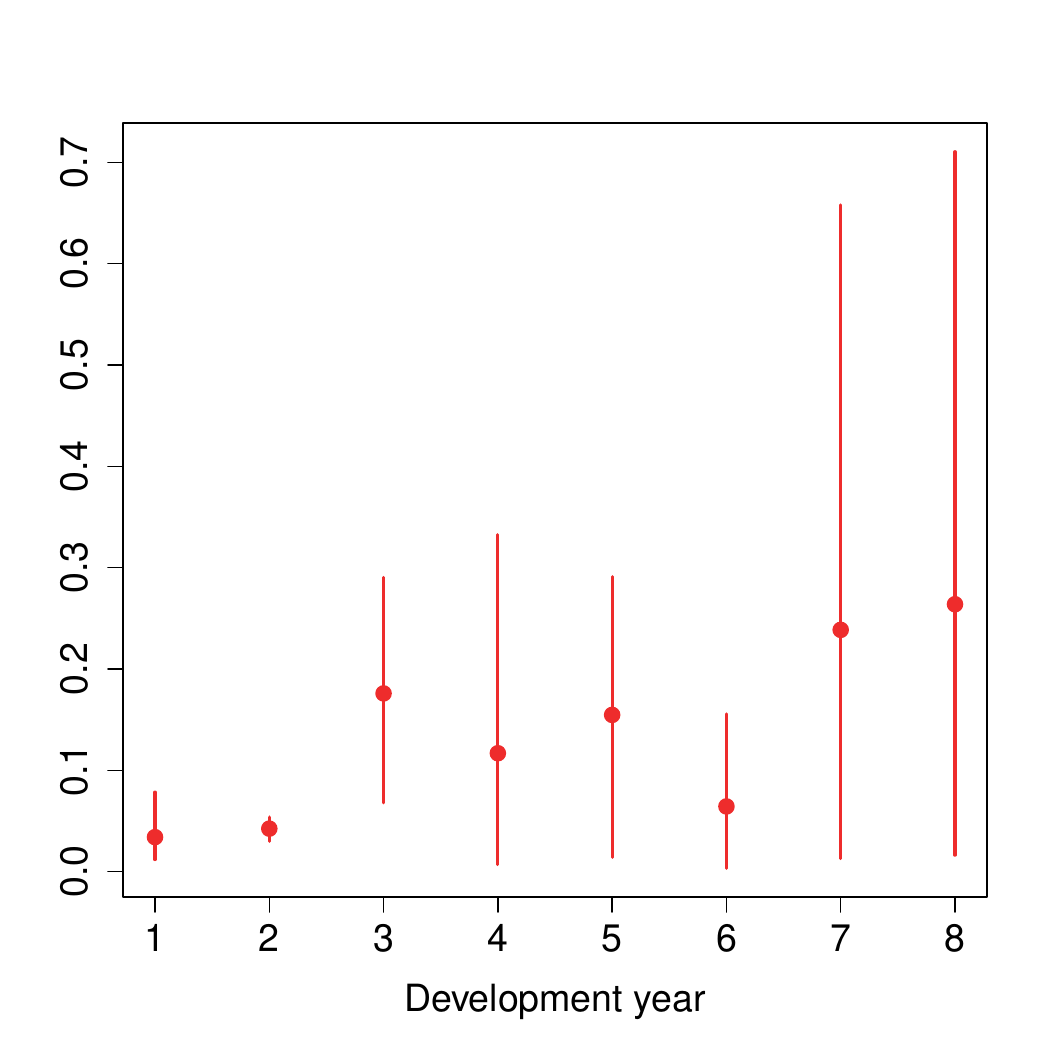}
\includegraphics[scale=0.45]{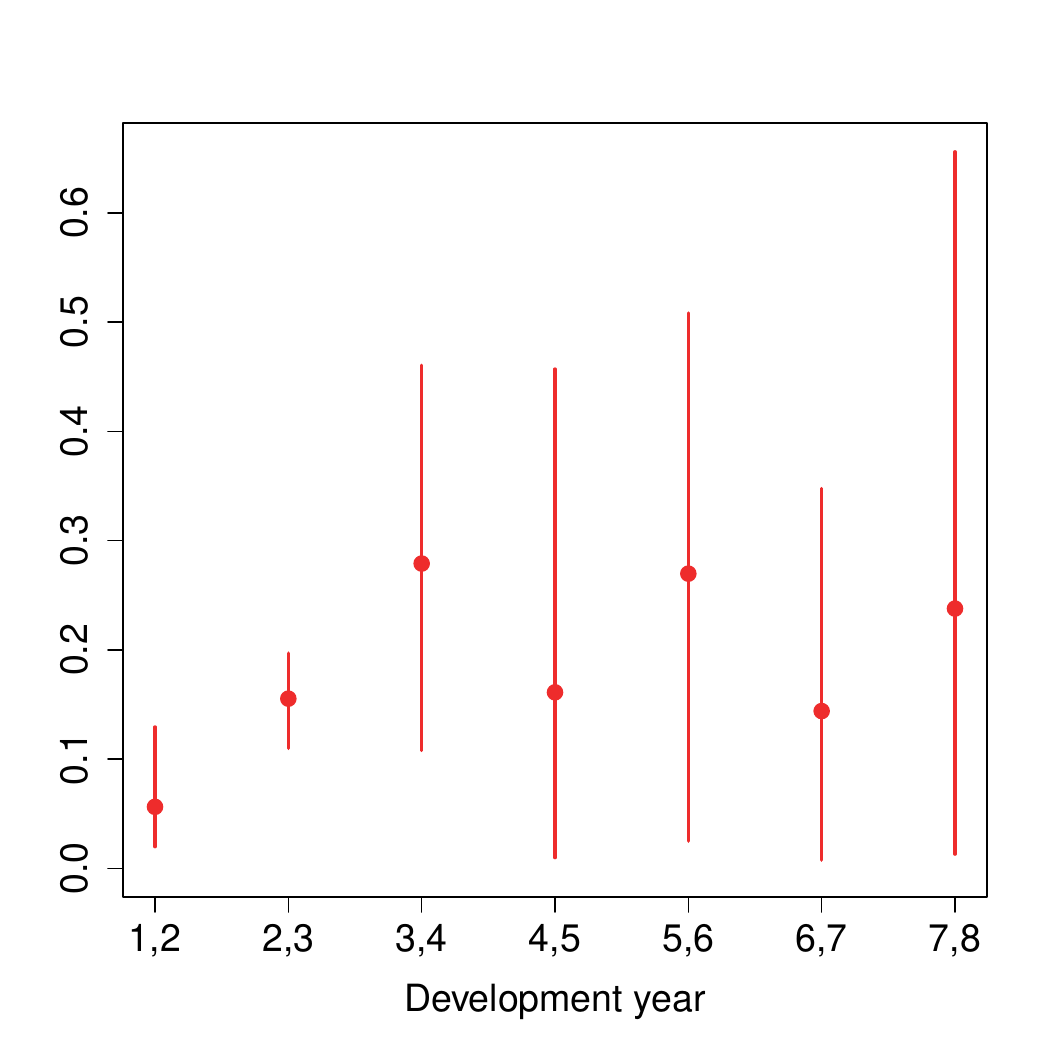}
\caption{Automobile data. Posterior estimates for parameters $\gamma_j$ (left) and $\rho_{j,j+1}$ (right) for $j=1,\ldots,n$ with $n=8$. Posterior mean (dots) and 95\% CI (lines).}
\label{fig:gammaa}
\end{figure}

\begin{figure}
\centering
\includegraphics[scale=0.45]{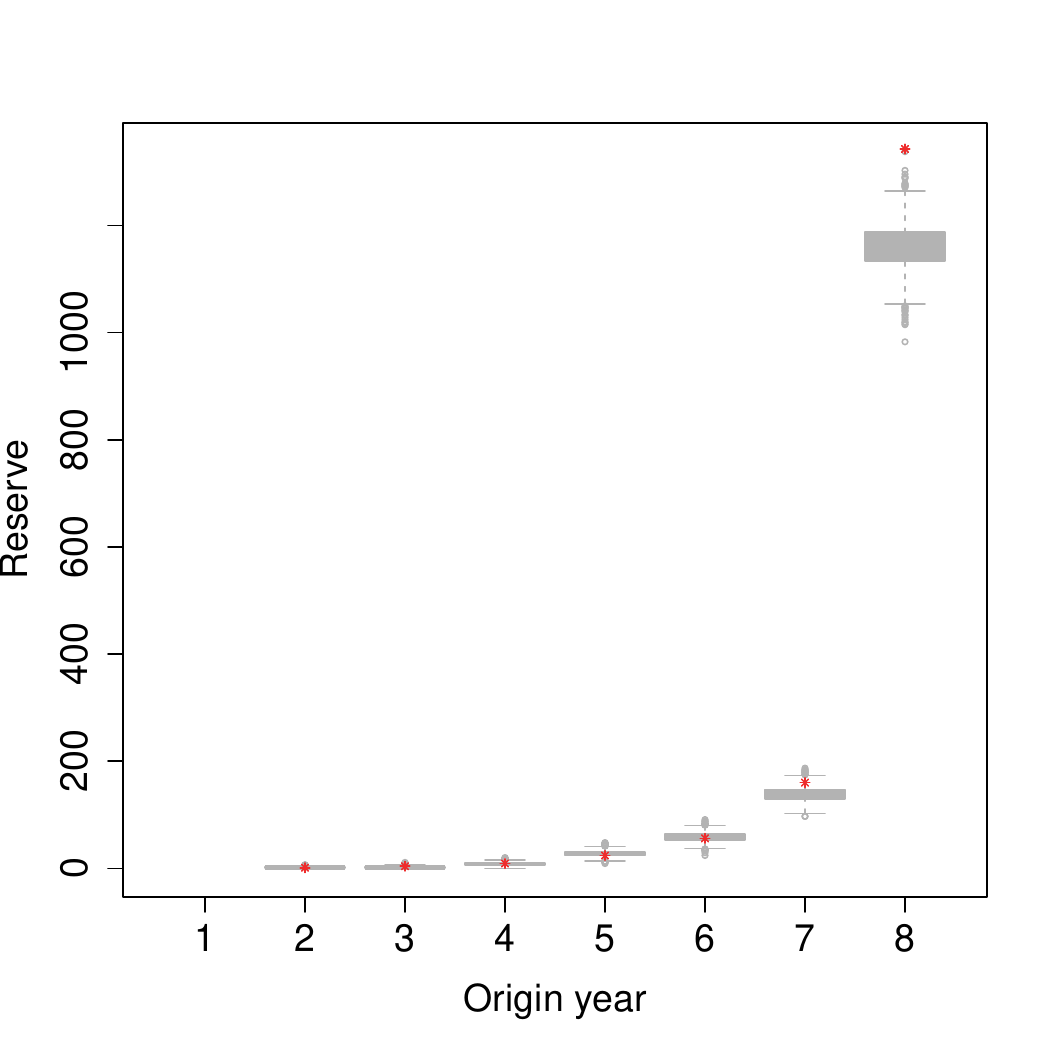}
\includegraphics[scale=0.45]{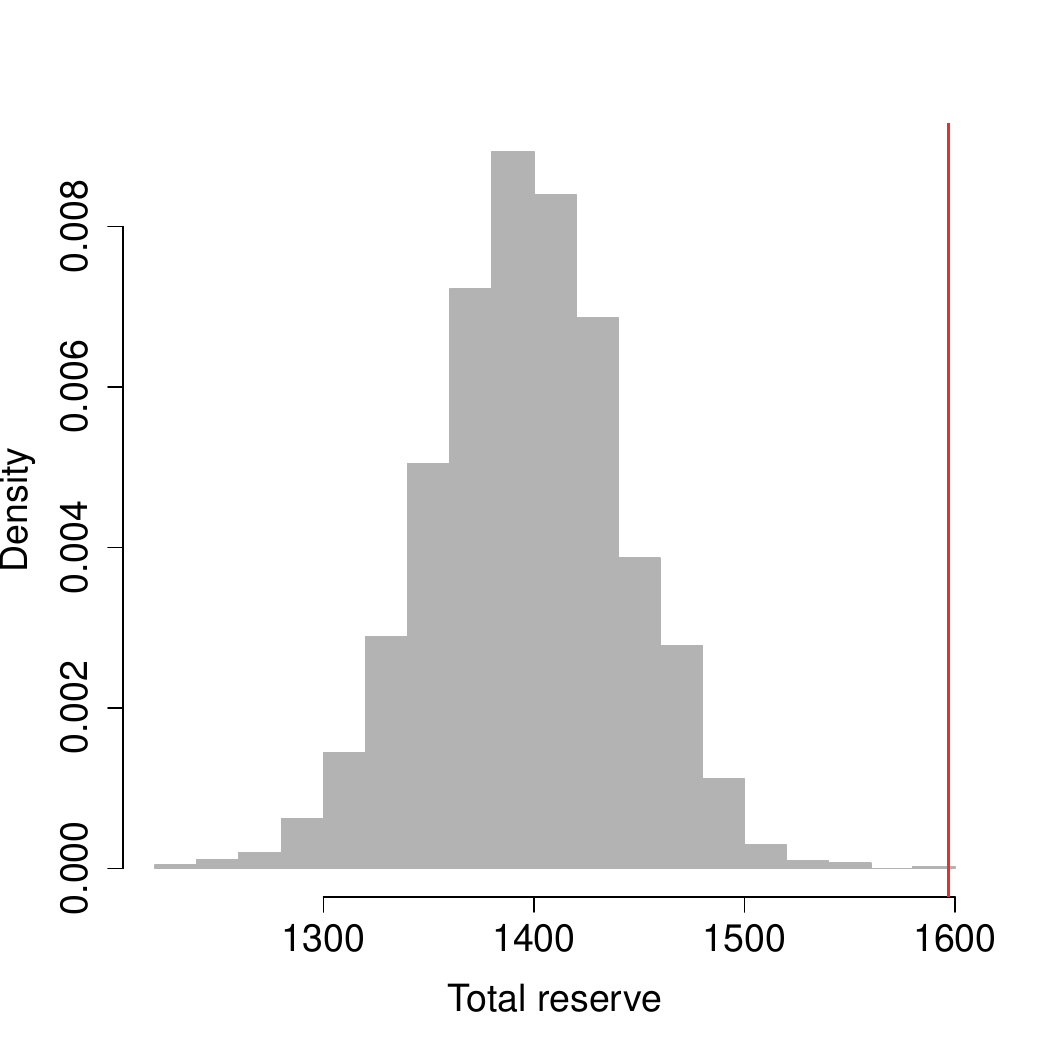}
\caption{Automobile data. Posterior predictive distributions of: aggregated number of claims $N_i$, $i=2,\ldots,n$ (left) and total aggregated number of claims $N$ (right). Histogram (grey) and chain ladder point estimates (asteriks).}
\label{fig:resa}
\end{figure}

\end{document}